\documentclass[11pt]{article}
\usepackage[margin=1in]{geometry}

\usepackage{authblk}

\usepackage{amsmath,amsfonts,amssymb,amsthm}

\def\final{1}

\usepackage[noend]{algpseudocode}
\usepackage{algorithm}

\usepackage{microtype}

\usepackage{accents}

\usepackage{verbatim}

\usepackage{paralist}

\usepackage{hyperref}
\usepackage{nameref}

\newcounter{mylabelcounter}

\makeatletter
\newcommand{\labelText}[2]{%
#1\refstepcounter{mylabelcounter}
\immediate\write\@auxout{%
  \string\newlabel{#2}{{1}{\thepage}{{\unexpanded{#1}}}{mylabelcounter.\number\value{mylabelcounter}}{}}
}%
}
\makeatother

\newenvironment{assumption}[3]{
\begin{center}
\begin{minipage}[c]{.95\textwidth}
\begin{tabular}{lp{0.84\textwidth}}
{\bf \labelText{(#1)}{#2}} & {#3} \\
\end{tabular}
\end{minipage}
\end{center}
}

\newcommand{\Z}{\mathbb{Z}}
\newcommand{\Q}{\mathbb{Q}}
\newcommand{\Vm}{V \setminus \{t\}}
\newcommand{\VN}{V^-}
\newcommand{\VP}{V^+}
\newcommand{\VZ}{V^0}

\DeclareMathOperator{\Ex}{Ex}
\DeclareMathOperator{\Deficit}{Def}

\newcommand{\netflow}[2]{\nabla #1_{#2}}

\newcommand{\Rinf}{\bar\R}

\makeatletter
\algnewcommand{\LineComment}[1]{\Statex \hskip\ALG@thistlm \(\triangleright\) \emph{#1}}
\makeatother

\newtheorem{theorem}{Theorem}[section]

\newtheorem{definition}[theorem]{Definition}

\newtheorem{lemma}[theorem]{Lemma}

\newtheorem{claim}[theorem]{Claim}
\newtheorem*{claim*}{Claim}

\ifdefined\final
    \newcommand{\todo}[1]{}
    \newcommand{\question}[1]{}
    \newcommand{\nnote}[1]{}
    \newcommand{\lnote}[1]{}
    
    \newcommand{\modified}[1]{#1}
    \newcommand{\edited}[1]{#1}
\else
\colorlet{darkred}{red!50!black}
    \newcommand{\todo}[1]{{\color{darkred}\em TODO: #1}}
    \newcommand{\question}[1]{{\color{darkblue}\emph{(#1)}}}
    \newcommand{\nnote}[1]{\begingroup\color{blue!50!black}\em (Neil: #1)\endgroup}
    \newcommand{\lnote}[1]{\begingroup\color{purple!50!black}\em (Laci: #1)\endgroup}
    
    \newcommand{\modified}[1]{{\color{green!45!black}#1}}
    \newcommand{\edited}[1]{{\color{green!45!black}#1}}
\fi

\newcommand{\RecursiveGenFlow}{{\sc Recursive-Generalized-Flow}}

\newcommand{\R}{\mathbb{R}}

\DeclareMathOperator{\supp}{supp}

\newcommand{\globpot}{\overline{\Ex}}

\newcommand{\apot}{\Phi}
\newcommand{\npot}{\Psi}

\newcommand{\fs}{\bar{f}}
\newcommand{\mus}{\bar{\mu}}

\newcommand{\Sc}{S}

\newcommand{\Tc}{T}
\newcommand{\cI}{\mathcal{I}}
\newcommand{\Vng}{W}
\newcommand{\Vg}{X}
\newcommand{\Vt}{Y}
\newcommand{\Vr}{Z}
\newcommand{\fg}{f^{\Vg}}
\newcommand{\ft}{f^{\Vt}}
\newcommand{\fr}{f^{\Vr}}
\newcommand{\mut}{\mu^Y}
\newcommand{\mur}{\mu^Z}
\newcommand{\fri}{\tilde{f}}
\newcommand{\muri}{\tilde{\mu}}

\newcommand{\aux}{o}
\newcommand{\sigmab}{\sigma_{\aux}}
\newcommand{\Vaux}{\hat{V}}
\newcommand{\Eaux}{\hat{E}}
\newcommand{\gammaaux}{\hat{\gamma}}
\newcommand{\gaux}[1]{\gammaaux_{\aux #1}}

\newcommand{\rpref}[1]{(R\ref{#1})}

\newcommand{\instance}{\mathcal{I}}

\newcommand{\ole}[1]{\accentset{\leftarrow}{#1}}
\newcommand{\obe}[1]{\accentset{\leftrightarrow}{#1}}

\newcommand{\rndgain}{\theta}

\begin{document}

\title{A Simpler and Faster Strongly Polynomial Algorithm for  Generalized Flow Maximization\thanks{A preliminary version of this paper has appeared in the \emph{Proceedings of the 49th Annual ACM SIGACT Symposium on Theory of Computing}, pages 100--111.}}
\author[1,2]{Neil Olver\thanks{Supported by an NWO Veni grant, and NWO Vidi grant 016.Vidi.189.087.}}
\author[1]{L\'aszl\'o A. V\'egh\thanks{Supported by EPSRC First Grant EP/M02797X/1, and by  ERC
  Starting Grant ScaleOpt--757481.}}
\affil[ ]{ \texttt{\{n.olver,l.vegh\}@lse.ac.uk}}
\affil[1]{London School of Economics and Political Science, London, UK}
\affil[2]{Centrum Wiskunde \& Informatica, Amsterdam, The Netherlands}

\maketitle

\begin{abstract}
    We present a new strongly polynomial algorithm for generalized flow maximization that is significantly simpler and faster than the previous strongly polynomial algorithm \cite{laci}.
\edited{For the uncapacitated problem formulation, the  complexity bound $O(mn(m+n\log n)\log (n^2/m))$ improves on the previous estimate by almost a factor $O(n^2)$. }
Even for small numerical parameter values, our running time bound is comparable to the best weakly polynomial algorithms.
The key new technical idea is relaxing the primal feasibility conditions.  This allows us to work almost exclusively with integral flows, in contrast to all previous algorithms for the problem.

\end{abstract}


\section{Introduction}
In the maximum generalized flow problem, we are given a directed graph $G=(V,E)$ with a sink node $t\in V$. Every arc $e\in E$ has a positive capacity and a positive gain factor $\gamma_e> 0$. 
Flow entering an arc $e$ gets rescaled by the factor $\gamma_e$ when traversing the arc;  
the goal is to maximize the amount of flow sent to the sink. 
The generalized flow model dates back to Kantorovich's seminal 1939 paper~\cite{kantorovich}, the same paper that formally introduced linear programming.

 Generalized flow networks can model transportation of a commodity through a network, such as a liquid or gas through pipelines, where loss is experienced.
 Nodes can also represent different types of entities that can be converted into each other at certain conversion rates. 
 For example, in financial networks, the nodes can represent various equities, and 
 arcs correspond to possible trades.
Generalized flows can also be used to model generalized assignment problems, such as assigning raw materials to final products where gain factors can encode the efficiency of the processing.
We refer the reader to Ahuja et al.~\cite[Chapter 15]{amo} for further applications of the model.

Early combinatorial algorithms were developed by Dantzig~\cite{dantzig63} and by Onaga~\cite{onaga66}.
 The first polynomial-time combinatorial algorithm was given by Goldberg, Plotkin, and Tardos~\cite{Goldberg91} in 1991.  A large number of weakly polynomial algorithms were developed in the subsequent 20 years~\cite{Cohen94,Fleischer02,Goldfarb96,Goldfarb02a,Goldfarb97,Goldfarb02,Kapoor96, Radzik04,Restrepo09,Tardos98,Vaidya89,Vegh11,Wayne02}.
Let $n$ denote the number of nodes,  $m$ the number of arcs of the graph, and let be $B$ be the largest integer in the description of the gain factors, capacities, and node demands.
Among the previous algorithms, the best  running times are the $O(m^{1.5}n^2\log (nB))$ interior point method  by Vaidya~\cite{Vaidya89}; and the $O(mn(m+n\log n)\log B)$ combinatorial algorithm by Radzik~\cite{Radzik04}.
Interior point methods can obtain fast approximate solutions for lossy networks, i.e. if $\gamma_e\le 1$ for all arcs.
The result of Daitch and Spielman~\cite{daitch} finds an additive $\varepsilon$-approximate solution in $\tilde O\left(m^{3/2}\log^2(B/\varepsilon)\right)$, recently improved by Lee and Sidford~\cite{lee2014} to 
$\tilde O\left(m\sqrt{n}\log^{O(1)}(B/\varepsilon)\right)$.\footnote{The notation $\tilde O(.)$ hides further polylog$(m)$ factors.}
 However, these results do not obtain an exact solution.

Resolving  a longstanding open question, the first strongly polynomial algorithm was given by V\'egh~\cite{laci}, with running time $O(n^3m^2)$ for the uncapacitated variant of the problem.\footnote{The problem can be formulated in several equivalent forms, as detailed in Section~\ref{sec:prelim}. The running times above were quoted for the standard capacitated form; in this form, \cite{laci} yields $O(m^5)$.}
The main progress in the algorithm is that, within a strongly polynomial number of steps, at least one arc can be identified that must be tight in every dual optimal solution. 
Consequently, the size of the instance can be reduced by contracting such arcs. 
The algorithm is based on \emph{continuous scaling}, a novel version of the classical scaling method. 
The algorithm is technically very complicated. 

We give a new algorithm for generalized flow maximization that improves on \cite{laci} both in terms of speed and simplicity.
Our new algorithm works along broadly similar lines, and also involves arc contractions as a main vehicle of progress, with path augmentation and relabelling operations being used to find an arc to contract.
At the same time, our algorithm introduces a number of new conceptual and technical ideas. 

We give a detailed technical overview and comparison at the beginning of Section~\ref{sec:alg}, after having defined the basic notation and concepts. 
Here we briefly highlight a key novelty.
Unlike all previous combinatorial algorithms, we do not maintain a feasible primal solution (i.e., flow).
Instead, we ensure that the dual solution has a certain property that keeps us ``within reach'' of a feasible primal solution that respects certain complementary slackness conditions.
So while our algorithm is a primal-dual algorithm, in a sense it does not keep track of the ``real'' primal but only a proxy for it.
Working with an infeasible primal solution turns out to have major benefits; in particular, we are able to work almost exclusively with integer flows, simplifying matters dramatically.

Our running time bound is $O(mn(m+n\log n)\log (n^2/m))$ for the uncapacitated form, a substantial improvement over V\'egh~\cite{laci}. 
For the capacitated form, we obtain $O(m^2(m+n\log n)\log m)$. \nnote{Reverted my $\log n$ change.}
For uncapacitated instances, our running time is better than the interior point method of Vaidya~\cite{Vaidya89} for arbitrary values of the complexity parameter $B$, and better than Radzik's combinatorial algorithm~\cite{Radzik04} if $B=\omega(n^2/m)$.

\paragraph{The context of strongly polynomial Linear Programming.}  
The currently known polynomial-time algorithms for Linear Programming (LP), such as the ellipsoid and interior point methods, are \emph{weakly polynomial}, as the bound on the number of arithmetic operations depends on the numerical input. 
In contrast, in a \emph{strongly polynomial} LP algorithm, the number of elementary arithmetic operations must be bounded polynomially in the number of variables and the number of constraints. 
Furthermore, the algorithm must be in PSPACE, that is, the numbers occurring in the computations must remain polynomially bounded in the input size.

Finding a strongly polynomial LP algorithm is a major open question: it was listed by Fields medalist Stephen Smale not only as the main unsolved problem in linear programming theory, but  as one of the most important challenges for mathematics in the twenty-first century~\cite{smale1998}.

Consider an LP in the following standard form, with $A\in \R^{n\times m}$, $b\in \R^n$, $c\in \R^m$. 
\begin{equation}\tag{LP}\label{LP}
\begin{aligned}
\min~& c^\top x\\
Ax & =b\\
x &\geq 0.
\end{aligned}
\end{equation}
The most general strongly polynomial computability results are due to Tardos~\cite{Tardos86}, and to Vavasis and Ye~\cite{vavasis96}. 
In these results, the running time only depends on the matrix $A$, but not on the right hand side $b$ or on the cost $c$. 
Tardos~\cite{Tardos86} assumes that $A$ is integer, and obtains a running time $\text{poly}(n, m, \log \Delta)$, where $\Delta$ is an upper bound on the largest subdeterminant of $A$.
In particular, if all entries in $A$ are integers of size poly$(n,m)$, this algorithm is strongly polynomial. 
These are called \emph{``combinatorial LPs''} since most network optimization problems can be expressed with small integer constraint matrices.
Vavasis and Ye~\cite{vavasis96} waive the integrality assumption, replacing $\Delta$ with a more general condition number.

A different, natural restriction on \eqref{LP} is to impose constraints on the nonzero elements. 
Assume that every column of the constraint matrix $A$ has at most two nonzero entries, but these can be arbitrary numbers. 
Let $\mathcal{M}_2(n,m)\subseteq \R^{n\times m}$ denote the set of all such matrices. 
The results~\cite{Tardos86,vavasis96} do not apply for LPs with such constraint matrices. 
It is easy to see that every LP can be equivalently transformed to one with at most three nonzeros per column.

For the dual feasibility problem, that is, finding a feasible solution to $A^\top y\ge c$ for $A\in \mathcal{M}_2(n,m)$, Megiddo~\cite{megiddo83} gave a strongly polynomial algorithm. \nnote{changed $(m,n)$ to $(n,m)$.}
In fact, the notion of strongly polynomial algorithms was formally defined in the same paper (called ``genuinely polynomial''). 
The primal feasibility problem, that is, finding a feasible solution to $Ax=b$, $x\geq0$ for $A\in \mathcal{M}_2(n,m)$, can be reduced to generalized flow maximization~\cite[Section 8]{laci}.
Hence, the algorithm by V\'egh~\cite{laci}, as well as our new algorithm, give strongly polynomial algorithms for primal feasibility.

It remains an important open question to solve the optimization \eqref{LP} for a constraint matrix $A\in \mathcal{M}_2(n,m)$ in strongly polynomial time. This problem reduces to the minimum cost generalized flow problem~\cite{Hochbaum04}. 
As our new algorithm gives a simple and clean solution to flow maximization, we expect that the ideas developed here bring us closer to resolving this problem.

\medskip

The rest of the paper is structured as follows. Section~\ref{sec:prelim} introduces notation, basic concepts and preliminary results. Section~\ref{sec:alg} describes the overall algorithm and its subroutines. Section~\ref{sec:analysis} presents the analysis of the main algorithm. Section~\ref{sec:first-phase} shows how the initial primal and dual solutions can be obtained.
We briefly discuss open directions in Section~\ref{sec:conclusion}.



\section{Problem and preliminaries}\label{sec:prelim}
Let $\R_+$ and $\R_{++}$ denote the nonnegative and positive reals respectively; similarly let $\Z_+$ and $\Z_{++}$ denote the nonnegative and positive integers.
Let $\Rinf=\R\cup \{\infty\}$, and similarly for other cases.
For a vector $x\in \R^I$, $\|x\|_p$ denotes its $p$-norm, and \edited{$\supp(x)=\{i\in I: x_i\neq 0\}$ denotes its support.}
For a $S\subseteq I$, we use $x(S)$ as shorthand for $\sum_{i \in S} x_i$.

Let $(V,E)$ be a simple directed graph, which we assume to be connected in an undirected sense.
Let $n:=|V|$ and $m:=|E|$.  
For an arc set $F\subseteq E$, let  $\ole{F}:=\{ji\colon ij\in F\}$ denote the
set of reversed arcs, and let $\obe{F}:=F\cup \ole{F}$.
For a subset $S\subseteq V$, we let $E[S]$ denote the set of arcs with both endpoints inside $S$. 
\edited{Further, we let $\delta^-_E(S)$ and $\delta^+_E(S)$ denote the set of incoming and outgoing arcs, respectively. If $E$ is clear from the context, we simply write $\delta^-(S)$ and $\delta^+(S)$.}
If $S=\{i\}$, we use the simplified notation $\delta^-(i)$ and $\delta^+(i)$, and also define $\delta(i) := \delta^+(i) \cup \delta^-(i)$.

An instance of the generalized flow problem is given as
$\instance=(V,E,t,\gamma,b)$, where
$t\in V$ is a sink node, $\gamma\in \R_{++}^E$ is the vector of gain factors, 
and $b\in\R^{V \setminus \{t\}}$ is the vector of node demands. 
Let us partition the nodes according to the sign of the demand.
\[
	\VN := \{ i \in \Vm : b_i < 0\}, \quad 
        \VZ := \{ i \in \Vm : b_i = 0\}, \quad 
        \VP := \{ i \in \Vm : b_i > 0\}.
\]
The \emph{net flow} at a node $i$ is defined as 
\[ \netflow{f}{i} := \sum_{e\in \delta^-(i)} \gamma_e f_e - \sum_{e\in \delta^+(i)} f_e. \]
We are ready to formulate the generalized flow maximization problem.
\begin{equation}\tag{P}\label{primal}
\begin{aligned}
 \max \quad & \netflow{f}{t} \\
 \text{s.t.}\quad \netflow{f}{i}&\geq b_i\quad \forall i\in\Vm\\
  f &\geq 0.
\end{aligned}
\end{equation}
The problem can be formulated in multiple equivalent variants. A
formulation commonly used in the literature adds arc capacities and
sets all node demands to zero. All these formulations can be efficiently reduced to
\eqref{primal}. Moreover, every LP in the form $Ax=b, x\geq0$ for $A\in
\mathcal{M}_2(n,m)$ reduces to \eqref{primal} (see \cite[Section 8]{laci}
for the reductions).
The special case when  $\gamma_e=1$ for all $e\in E$ corresponds to
the standard network flow model; we will refer to standard network
flows as \emph{regular flows} to distinguish them from generalized
flows.
The dual program can be equivalently written in the 
following form. 
\begin{equation}\tag{D}\label{dual}
\begin{aligned}
\max\quad & \mu_t\sum_{j\in \Vm}\frac{b_j}{\mu_j}\\
 \text{s.t.}\quad \mu_j&\ge \gamma_{ij} \mu_i \quad \forall  ij \in E\\
  \mu_t&\in \R_{++}\\
 \mu_i&\in\Rinf_{++}\quad \forall i\in\Vm.
\end{aligned}
\end{equation}
The dual variable for node $i$ would be
$\mu_t/\mu_i$. Nodes other than $t$ are allowed to have
$\mu_i=\infty$; this corresponds to dual values 0. 
We are ready to state our main result.
\begin{theorem}\label{thm:main}
There exists a strongly polynomial algorithm, that, for any input
instance $\mathcal{I}=(V,E,t,\gamma,b)$, \modified{decides whether  \eqref{primal} is feasible and bounded, and if that is the case, then returns optimal
solutions to \eqref{primal} and \eqref{dual}. 
The algorithm runs in
 $O(mn(m+n\log n)\log(n^2/m))$ arithmetic operations.}
\end{theorem}
\lnote{The statement ignored the infeasible and unbounded cases. Is there a more compact way of writing this? Do we want to explicitly give Farkas certificates and infinite rays?  These should be easy to extract.}
Consider the setting where capacity constraints $f\le u$ are added to the
formulation \eqref{primal}. Such a problem can be reduced to the
uncapacitated form \eqref{primal} on an extended network with $n'=n+m$ nodes
and $m'=2m$ arcs, where a new node is added for every arc of the
original network; see \cite[Section 8]{laci} for the precise
reduction. 
\modified{Directly applying Theorem~\ref{thm:main} to this extended network would yield a running time of $O(m^2(m+m\log m)\log m)$ for the capacitated form.
    Because of the specific form of the extended network, the running time of a key shortest path computation can be improved from $O(m + m \log m)$ to $O(m+n \log n)$, as was done by Orlin~\cite{Orlin93}.
    This yields a running time of $O(m^2(m+n\log n)\log m)$.
}
\lnote{What happened with the final $\log m$, why did that also become $\log n$?}
\nnote{Corrected.}

\paragraph{Relabellings.}
We interpret the dual solutions as \emph{relabellings}, the basic vehicle of our algorithm. 
This is a standard technique used in the vast majority of generalized
flow algorithms.\footnote{Relabellings are inverse dual variables
  instead of the dual variables. This usage was introduced by
  \cite{Goldberg91} and has been  the standard formalism in the subsequent literature;
  we adhere to this inverse notation.}
A feasible solution $\mu \in \Rinf_{++}^V$ to \eqref{dual} is called a \emph{feasible labelling}. 
We define 
\[
    f^\mu_{ij} := \frac{f_{ij}}{\mu_i}\quad \forall ij\in E.
\]
The multiplier $\mu_i$ can be interpreted as a change of the unit of measurement at node $i$. 
An equivalent problem instance is obtained by defining
\[ 
    \gamma_{ij}^\mu := \gamma_{ij} \cdot \frac{\mu_i}{\mu_j}, 
        \qquad \qquad
    \netflow{f^\mu}{i}:=\frac{\netflow{f}{i}}{\mu_i}, 
        \qquad \text{and} \qquad 
    b_i^\mu:= \frac{b_i}{\mu_i}. 
\]
We use the convention $\gamma_{ij}^\mu=1$ if 
 $\mu_i=\mu_j=\infty$. 
Then the feasibility of $\mu$ to \eqref{dual} is equivalent to $\gamma^\mu_e \leq1$ for all $e \in E$.
We call an arc $e\in E$ \emph{tight} with respect to $\mu$, if
$\gamma^\mu_e=1$. Let $E^\mu$ and $\obe E^\mu$ denote the set of tight arcs
for $\mu$ in $E$ and in $\obe E$, respectively.

For a flow $f\in \R_{+}^E$, we define the residual graph $(V,E_f)$
with $E_f=E\cup\{ji: ij\in E, f_{ij}>0\}$. 
The latter set of arcs are called \emph{reverse arcs}. 
For a reverse arc $ji$, we define 
$\gamma_{ji}:=1/\gamma_{ij}$.
By increasing (decreasing) $f_{ji}$ by $\alpha$ on a reverse arc $ji\in E_f$, 
we mean decreasing (increasing) $f_{ij}$ by $\alpha/\gamma_{ij}$.

Let us define the \emph{excess} of a node $i$ under $f^\mu$ to be the amount $\nabla f^\mu_i - b^\mu_i$; 
so $f^\mu$ is feasible if all nodes have nonnegative excess.
We also define the total (positive) excess and the total deficit of $f^\mu$ as 
\[
    \Ex(f,\mu):=\sum_{i\in \Vm}\max\left\{\netflow{f^\mu}{i}-b_i^\mu,0\right\} \quad \text{and} \quad
\Deficit(f,\mu):=\sum_{i\in   \Vm}\max\left\{b_i^\mu -\netflow{f^\mu}{i},0\right\}. 
\]
For an arc set
$F\subseteq \obe{E}$, we let $\gamma(F):=\prod_{e\in F} \gamma_e$; $\gamma^\mu(F)$ is
defined similarly. \footnote{This is inconsistent with the shorthand notation $x(S)=\sum_{i \in S} x_i$. The multiplicative convention will be used for gain factors, and the additive one in all other contexts.}
Note that for any cycle $C$, and any finite
labelling $\mu$, we have $\gamma^\mu(C)=\gamma(C)$.

\edited{A {\em path flow} $h$ is a nonnegative vector $h\in \R_{+}^E$ such that $\supp(h)$ is a simple directed path. If $p$ is the starting point and $q$ is the endpoint of the path, then $\netflow{h}{p}\le 0$, $\netflow{h}{q}\ge 0$, and $\netflow{h}{i}= 0$ for $i\in V\setminus \{p,q\}$.}

Given two flows $f$ and $g$, the {\em difference} $f - g$ is the flow on $\obe{E}$ defined by the following.
For each $ij \in E$ where $f_{ij} \geq g_{ij}$, we set $(f-g)_{ij} = f_{ij} - g_{ij}$ and $(f-g)_{ji} = 0$.
For each $ij \in E$ where $f_{ij} < g_{ij}$, we set $(f-g)_{ij} = 0$ and $(f-g)_{ji} = (g_{ij} - f_{ij})/\gamma_{ij}$.

\paragraph{Fitting pairs and optimality.}
\begin{definition}\label{def:fitting}
Let $f\in \R_{+}^E$ and $\mu\in \R^V_{++}$. 
Then $(f,\mu)$ is called a \emph{fitting pair}, 
if $\mu$ is feasible to \eqref{dual}, and $f_e>0$ implies $\gamma^\mu_e=1$. 
\end{definition}
We also say that $f$ \emph{fits} $\mu$, or $\mu$ \emph{fits} $f$. 
Fitting captures a complementary slackness property. It is worth
noting that for a fitting pair $(f,\mu)$, $f^\mu$ is a regular flow.
The definition is equivalent to saying that $\gamma^\mu_e\le 1$ for
every $e\in E_f$.

The definition requires $\mu$ to be finite and feasible to \eqref{dual}, 
but not the feasibility of $f$ to \eqref{primal}. 
In fact, we will allow flows $f\in\R_{+}^E$ in the algorithm that violate the node balance constraints in \eqref{primal}.
In general, there may not exist a finite optimal solution to 
\eqref{dual}.
The following notion allows us to nevertheless work with finite dual
solutions only.
\begin{definition}
    A fitting pair $(f,\mu)$ is called \emph{essentially optimal} if
    $f$ is feasible and $\nabla f_i  = b_i$ for all $i \in V$ that can
    reach $t$ in $E_f$. 
The dual solution $\mu$ is \emph{essentially optimal} if there exists
a flow $f$ such that $(f,\mu)$ is essentially optimal.
\end{definition}
\edited{For an essentially optimal pair $(f,\mu)$, the primal solution $f$ is also optimal. While an essentially optimal dual solution $\mu$ need not itself be optimal, it provides enough information to easily and quickly obtain both an optimal primal and an optimal dual, as stated in the following lemma.
}
\lnote{Reorganised and moved forward comments.}
\begin{lemma}\label{lem:optimality}
Let $\mu$ be an essentially optimal dual solution. Then we can obtain
a flow $f$ such that  $(f, \mu)$ is an essentially optimal fitting
pair via a maximum flow computation. If $(f,\mu)$ is essentially
optimal, then $f$ is an optimal solution to \eqref{primal},
and $\mu'$ is an optimal solution to \eqref{dual}, where $\mu'$ is
defined by $\mu'_i:=\mu_i$ if $i$ can reach $t$ in $E_f$, and $\mu'_i :=
\infty$ otherwise.
\end{lemma}
\begin{proof}
For the first part, consider the graph $(V,E^\mu)$ of tight arcs. We
obtain the network $(V',E')$ by adding a new source node $s$ with a
new arc $si\in E'$ whenever $b_i<0$,
and a new arc $is\in E'$ whenever $b_i>0$. We set lower capacity 0 and upper
capacity $\infty$ for all arcs in $E^\mu$. We set lower capacity 0 and
upper capacity $-b_i^\mu$ for every arc $si\in E'$, and lower capacity
$b_i^\mu$ and upper capacity $\infty$ for every arc $is\in E'$. We
compute a maximum $s$-$t$ flow $h$ in this network, where $t$ is the
sink of the instance. 
Then it is easy to see that if $(f,\mu)$ is
essentially optimal, then $h=f^\mu$ (with a natural extension to the
$is$ and $si$ arcs) is a maximum $s$-$t$ flow in this
network. Conversely, for any maximum $s$-$t$ flow $h$, if we define $f$  as
$f_{ij}=h_{ij}\mu_i$ for every $ij\in E$, then $(f,\mu)$ is an
essentially optimal fitting pair.
 The second part of the lemma is immediate
by complementary slackness.  
\end{proof}
\modified{
Consequently, an essentially optimal dual solution contains more information than an optimal dual solution. Given an optimal dual  solution with infinite values, there is no easy way of obtaining an optimal primal solution. This requires finding a feasible generalized flow on the 
subgraph induced by the nodes with infinite dual optimum value, a problem essentially equivalent to the maximum generalized flow problem.
}

The following property will be crucial in our algorithm. 
\begin{definition}\label{def:safe}
    We say that a feasible labelling $\mu \in \R_{++}^V$ is \emph{safe} if there exists a feasible
solution $f$ to \eqref{primal} such that $(f,\mu)$ is a fitting pair. 
\end{definition}
Safe labellings can be characterized using Hoffman's circulation theorem.
\begin{lemma}\label{lem:safe-cut}
The feasible labelling $\mu \in \R_{++}^V$ is safe if and only if
$b^\mu(Z)\le 0$ for every set $Z\subseteq \Vm$ with $\delta^-_{E^\mu}(Z)=\emptyset$.
\end{lemma}

We will make use of safe labellings via the next lemma.
\begin{lemma}\label{lem:safe-flow}
    Let $(f,\mu)$ be a fitting pair, and assume $\mu$ is safe. 
    Then there exists a feasible flow $g\in \R_+^E$ fitting $\mu$ with
    $\Ex(g,\mu)\le \Ex(f,\mu)$, and $\|f^\mu-g^\mu\|_\infty\le \Deficit(f,\mu)$.
\end{lemma}
\begin{proof}
    Let $g$ be a flow with $\supp(g) \subseteq E^\mu$ and $\nabla g^\mu_i \geq b^\mu_i$ for all $i \in \Vm$, chosen so that $\|f^\mu - g^\mu\|_1$ is as small as possible.
    The safety of $\mu$ implies that such a $g$ can be found.
    We show that $g$ satisfies the requirements of the lemma.

Consider the flow $h := f - g$; or in other words, we consider $h^\mu = f^\mu - g^\mu$, the difference of two regular flows.
Let $H$ denote the support of $h$; then $H\subseteq E_{g}$ and $\ole{H} \subseteq E_{f}$. 

We first note that if $H$ contained a directed cycle, then it would be possible to add some positive amount of this cycle to $g^\mu$ and obtain another feasible flow, contradicting our choice
of $g^\mu$ to minimize $\|f^\mu-g^\mu\|_1$.

Thus $h^\mu$ can be decomposed into a sum of path flows, 
where each path flow begins at a node $p$ with $\nabla f^\mu_p < \nabla g^\mu_p$ and 
ends at a node $q$ with $\nabla f^\mu_q > \nabla g^\mu_q$. This follows, e.g., by \cite{Gondran84}, see also \cite[Theorem 2.6]{Goldberg91}.
If the head $p$ of one of these paths were equal to $t$, or if it satisfied $\nabla g^\mu_p > b^\mu_p$, then we would be able to increase $g^\mu$ by some positive amount along this path, remaining feasible, and decreasing $\|f^\mu - g^\mu\|_1$.
It follows that the total flow carried over all the paths is exactly $\Deficit(f,\mu)$, and hence $\|h^\mu\|_\infty \leq \Deficit(f,\mu)$.
Since $\nabla f^\mu_p \geq \nabla g^\mu_p$ for any node where $\nabla f^\mu_p \geq b^\mu_p$, it follows that $\Ex(g,\mu) \leq \Ex(f,\mu)$.
\end{proof}

\paragraph{Contraction operation.}
The main progress during the algorithm will reduce the instance via contractions.
Given an instance $\instance=(V,E,t,\gamma,b)$ and an arc $e=pq\in E$,
the contracted instance  $\instance/e = (V', E', t', \gamma', b')$ is defined as follows.
\begin{itemize}[-]
    \item $(V',E') = (V,E) / e$ is the simple directed graph obtained
      from contracting $e$. If the contraction creates parallel arcs,
      we only keep one of them. 
        We retain the label $q$ to refer to the contracted image of $\{p,q\}$, so $V' = V \setminus \{p\}$. 
    \item For $ij \in E'$ where $i,j \neq q$, $\gamma'_{ij} = \gamma_{ij}$.
        For $iq \in E'$, we set $\gamma'_{iq}$ either to $\gamma_{iq}$ or $\gamma_{ip}\gamma_{pq}$, 
        depending on whether $ip$ or $iq$ is present in $E$.
        If both exist, we choose the \emph{larger} value \emph{(recall that arcs are uncapacitated; given two parallel arcs of different gains, there is no benefit to sending flow on one with smaller gain)}. 
    \item $b'_i = b_i$ for all $i \in V' \setminus \{q\}$, and $b'_{q} = b_q + \gamma_e b_p$.
    \item $t' = t$, unless $t \in \{p,q\}$, in which case $t' = q$.
\end{itemize}
Note that  contracting $e=pq$ is equivalent to restricting ourselves
to dual solutions in $\instance$ for which the arc $pq$ is
tight \edited{($\mu_p = \mu_q / \gamma_{pq}$), and replacing $\mu_p$ everywhere by $\mu_q / \gamma_{pq}$}.
Assume now we are given a fitting pair $(g,\mu)$ in $\instance$, where
$e=pq$ is tight for $\mu$. We can define
define a natural image $(g/e, \mu/e)$ of $(g,\mu)$ in $\instance / e$,
as follows. We use the convention that $g_{ij}=0$ for all pairs
$(i,j)$ with  $ij\notin E$.
\begin{itemize}[-]
    \item $\mu / e$ is simply the restriction of $\mu$ to $V' = V \setminus \{p\}$. 
    \item 
        For each arc $ij\in E'$ of the contracted instance, 
        \[
            (g/e)_{ij} := \begin{cases} g_{ij} &\text{ if } q \notin \{i,j\}\\
                 g_{iq} + g_{iq}/\gamma_{ip} &\text{ if } j=q\\
                 g_{qj} + g_{qj}\gamma_{pj} &\text{ if } i=\modified{q}
        \end{cases}
    \]
\end{itemize}
The following lemma is immediate due to $g^\mu_e = 1$.
\begin{lemma}\label{lem:contractadditive}
Let $(g,\mu)$ be a fitting pair in $\instance$, where
$e=pq$ is tight for $\mu$. Then
$\nabla (g/e)^{\mu/e}_{p} = \nabla g^{\mu}_p + \nabla g^{\mu}_q$, and if $t \notin \{p,q\}$, 
$b^{\prime \mu/e}_q = b^\mu_p + b^\mu_q$.
\nnote{Technically, this notation for $b$ in the contracted instance is imprecise. But it is clear. 
Should we change it to the uglier $b^{\prime \mu/e}_q$?}\lnote{Changed to the uglier yet precise form.}
\end{lemma}

\paragraph{Initial solutions.} 
Our main algorithm will use the following assumption.
\assumption{$\star$}{cond:init}{An initial fitting pair $(\bar
  f,\bar \mu)$ is given, where $\bar f\in \R_{+}^E$ is feasible to
  \eqref{primal}. 
  }
According to this assumption, both \eqref{primal} and
\eqref{dual} are feasible; in particular, \eqref{primal} is bounded.
For arbitrary input instances  (that can also be infeasible or
unbounded) we will use an overall scheme akin to the two phase simplex
method; this is described in Section~\ref{sec:first-phase}.
When calling the main algorithm in the first phase, it will be applied
to a modified instance where one can trivially
provide an initial fitting pair satisfying \nameref{cond:init}.
For the original problem, the first phase will either find an initial fitting pair satisfying
\nameref{cond:init}, or otherwise a certificate of infeasibility or
unboundedness.



\section{The generalized flow algorithm}\label{sec:alg}
\subsection{Technical overview}\label{sec:overview}
The main progress in the algorithm will be finding a fitting pair
$(f,\mu)$ with
a safe labelling $\mu$ containing an \emph{abundant arc}---an arc whose relabelled flow is large compared to the total excess and deficit across all nodes.
It can be shown that an abundant arc must be tight in every dual
optimal solution. Namely, if we contract an abundant arc, then an essentially optimal dual solution can be easily
pulled back from the contracted instance to the original one (see Lemma~\ref{lem:pullback} below). Once we
have an essentially optimal dual solution, primal and dual
optimal solutions can be found via Lemma~\ref{lem:optimality}.

This is very similar to the scheme used by V\'egh~\cite{laci}. In
fact, both algorithms can be seen as descendants of Orlin's algorithm~\cite{Orlin93}
for minimum cost circulations.
Orlin's algorithm repeatedly finds abundant arcs using a variant of the classical
Edmonds-Karp scaling algorithm~\cite{Edmonds72}. 
The size of the network can be reduced by contracting the abundant arcs. 
We note that the idea of obtaining strongly polynomial algorithms by
repeatedly identifying constraints that must be tight in every optimal
solution goes back to the seminal work by Tardos~\cite{Tardos85strong}.

\medskip

Our augmenting path subroutine for identifying abundant arcs is vastly
simpler and more efficient than the one in \cite{laci}. The crucial idea is
\emph{relaxing the feasibility of the flow $f$} in the augmenting path algorithm. That is,
nodes $i$ with $\netflow{f}{i}<b_i$ will be allowed. This is
a quite radical change compared to all previous generalized flow algorithms. In
fact, ``fixing'' a node deficit can be very difficult:
compensating for just a tiny shortfall in node demands can be at the
expense of an arbitrarily large drop in the objective value.
We avoid this problem by maintaining that the labelling $\mu$ remains safe throughout.

Relaxing feasibility enables a more natural algorithmic framework, and
eliminates some significant technical challenges
in previous algorithms.
 We need to maintain the safety of the
labelling, but this happens without additional effort.
The most salient consequences are the
following.
\begin{itemize}
\item First of all, we can easily \emph{maintain a fitting pair
$(f,\mu)$} throughout. In contrast, V\'egh~\cite{laci} had to introduce a
relaxation of this concept called $\Delta$-feasibility, depending on
the current scaling factor $\Delta$. An earlier algorithm that
maintained a fitting pair throughout was the algorithm of Goldfarb, Jin,
and Orlin \cite{Goldfarb97}, however, it came at the expense of
maintaining arc imbalances in an intricate bookkeeping framework.
\item Although our algorithm can be seen as an enhanced version of the
    \emph{continuous scaling} technique in V\'egh~\cite{laci}, the description
  does not even include a scaling factor, prevalent in the previous
  combinatorial methods. Instead,  we  maintain that {\em the
  relabelled flow $f^\mu$ is integral} throughout, except for the very
  final step when an exact optimum is computed. This has no precedent in previous algorithms, 
  and is surprising because the generalized flow problem is perceived as a genuinely non-integral problem. 
  Let us note that the
  value of $\mu_t$ corresponds to the scaling factor $\Delta$ used in
  previous scaling methods, e.g.,~\cite{Goldberg91,Goldfarb97,Radzik04,Vegh11,laci}; 
  we relax the standard requirement $\mu_t=1$ so that we can work with integer solutions.

 \item 
     A primary reason for the improved running time is a new and very direct \emph{additive} potential analysis, compared to the \emph{multiplicative} analysis in~\cite{laci}. 
     Roughly speaking, every path augmentation decreases our chosen potential by one; however, as long as there are no abundant arcs, the potential remains bounded.
     The analysis in \cite{laci} also charges the number of augmentations against a similar potential, but argues about the cumulative decrease in the scaling factor $\Delta$ in a rather indirect way.
\end{itemize}

In a strongly polynomial algorithm, one also needs to guarantee that the
sizes of numbers remain polynomially bounded in the input size. 
In the previous algorithm~\cite{laci}, this required cumbersome additional rounding steps. 
In contrast, we can achieve this rather easily.

Another technical novelty is our use of \emph{essentially optimal} dual solutions rather than \emph{optimal} dual solutions. 
An optimal dual solution may take on infinite values, 
in which case it may not provide enough information to find
a corresponding primal optimal solution via complementary slackness.
In most of the previous literature,
``dummy arcs'' of very small gains were added from all nodes to $t$ to enforce the
existence of a finite dual optimum. 
This is somewhat unattractive, and would be particularly problematic for an actual implementation due to numerical issues.
We circumvent this problem by using essentially optimal duals that always take on finite values, 
and contain sufficient information to easily identify primal and dual optima.

A further distinguishing feature of our algorithm is that we do not use an
initial cycle cancelling subroutine. Most combinatorial methods start
with the assuming the existence of an initial fitting pair as in
\nameref{cond:init}. In order to obtain this, flow generating cycles
(that is, cycles $C\subseteq E_f$ with $\gamma(C)>1$)
have to be eliminated first. Radzik~\cite{radzik93} adapted the
Goldberg-Tarjan minimum-mean cycle cancelling algorithm~\cite{Goldberg89} 
to cancel all flow generating cycles in strongly
polynomial time. We avoid using this subroutine, and instead perform
our algorithm in two phases, as in the two phase simplex
algorithm. In the first phase for feasibility, we obtain the fitting
pair used as the starting for the second phase. We note that the overall
running time of our algorithm is better than the running time of
Radzik's cycle cancelling subroutine~\cite{radzik93}. 
However, this two-phase scheme is not
particular to our current algorithm---it could be applied to previous algorithms as well.

\subsection{Initial rounding}\label{sec:round}
Our main algorithm is recursive in nature.
Rather than satisfying \nameref{cond:init}, this procedure (which we will refer to as \RecursiveGenFlow) will require that the input satisfies the following property.
\assumption{$\star\star$}{cond:integral}
{
    A fitting pair $(f,\mu)$ is given where $\mu$ is safe, $f^\mu$ is integral (but not necessarily feasible), and the support of $f$ is an orientation of a forest. 
}
In fact, we will maintain this property through the main steps of our algorithm.
The requirements that $\mu$ is safe and $f^\mu$ be integral will be crucial. 
On the other hand, the requirement that the support of $f$ be an orientation of a forest---for brevity, we will say that $f$ is \emph{acyclic} if this is satisfied---is  mild:
given any fitting pair $(f,\mu)$, it is easy to modify $f$ so that this holds.
The purpose of this requirement is to maintain a level of sparsity that will encourage the formation of abundant arcs and aid us in obtaining our desired running time. 

We apply a simple rounding procedure to transform an instance satisfying \nameref{cond:init} to one satisfying \nameref{cond:integral}, where in addition
$-1 < \nabla f^\mu_i - b^\mu_i < 2$ for all $i \in \Vm$.
This ensures that the total relabelled excess and deficit are both small, 
which will be important in obtaining a bound on the running time.
The transformation is based on the following lemma.
\begin{lemma}\label{lem:round}
    Let $(\bar{f},\mu)$ be a fitting pair. Then there exists an acyclic $f$ fitting $\mu$ such that $f^\mu\in \Z_+^E$, and $\lfloor \netflow{\bar{f}^\mu}{i}\rfloor \le
    \netflow{f^\mu}{i}\le \lceil \netflow{\bar{f}^\mu}{i}\rceil$.
Moreover, $f$ can be found with a single maximum flow computation.
\end{lemma}
\begin{proof}
Consider the feasible circulation problem on
$(V,E^\mu)$, with lower and upper node demands $\lfloor
\netflow{\bar{f}^\mu}{i}\rfloor$ and $\lceil \netflow{\bar{f}^\mu}{i}\rceil$.
The flow $\bar{f}^\mu$ is a feasible solution; hence, there exists an integer
solution $g$, which can be found by a maximum flow
algorithm. 
Since there are no arc capacities, $g$ can clearly be chosen to be acyclic.
Then $f_{ij}:= g_{ij} \mu_i$ is the desired solution.
\end{proof}
\modified{The first two steps of Algorithm~\ref{alg:roundmaxflow} implement the} reduction from \nameref{cond:init} to \nameref{cond:integral}.
Note that since $\mu$ is just a scaling of $\bar{\mu}$, and $\bar{f}$ is a feasible flow fitting $\bar{\mu}$, $\mu$ is clearly safe.
The final subroutine \textsc{Final-Solution}$(\hat \mu)$ takes an
essentially optimal dual $\hat\mu$, and computes optimal primal and
dual solutions as in Lemma~\ref{lem:optimality}. Namely, the primal
optimal solution $f$ can be obtained via a maximum flow computation, and
the dual optimal solution can be obtained by setting the dual
variables to infinity for the nodes that cannot reach $t$ in $E_f$.

\begin{algorithm}
    \caption{\sc Max-Generalized-Flow}\label{alg:roundmaxflow}
    \raggedright
    \begin{algorithmic}[1]
        \Require{Instance $\instance = (V,E,t,\gamma,b)$, fitting pair $(\bar{f},\bar{\mu})$ satisfying \nameref{cond:init}.}
    \Ensure{An optimal solutions to the systems \eqref{primal} and \eqref{dual}.} 
        \State $\mu \gets \bar{\mu} \cdot \max_{i \in \Vm} (\netflow{\bar f^{\bar\mu}}{i}-b_i^{\bar{\mu}})$.
        \State Round $\bar{f}$ using Lemma~\ref{lem:round} to obtain an acyclic $f$ fitting $\mu$ with $f^\mu$ integral and 
        $\lfloor b^\mu_i \rfloor \leq \netflow{f^\mu}{i} \leq \lceil b^\mu_i \rceil$ + 1.
        \State $\hat\mu\gets$ \Call{\RecursiveGenFlow}{$\instance,
          f,\mu$}.
\State \Return \Call{Final-Solution}{$\hat\mu$}
\end{algorithmic}
\end{algorithm}

\subsection{The overall algorithm via arc contraction}

We can now describe the overall structure of
\RecursiveGenFlow$(\instance,f,\mu)$, which is based on arc
contractions. We now give a sufficient condition to identify an arc
$e$ such that an essentially optimal pair in $\instance/e$ can be
extended to an essentially optimal pair in $\instance$.
\begin{definition}
    Suppose $(g,\mu)$ is a fitting pair with $\mu$ safe. 
    Then we call an arc $e \in E$ \emph{abundant} with respect to $(g,\mu)$ if 
    \[ g^\mu_e \geq \Ex(g,\mu) + \Deficit(g,\mu). \]
\end{definition}
\begin{lemma}\label{lem:pullback}
    Let $(g, \mu)$ be a fitting pair in $\instance$ with $\mu$ safe, and 
    suppose that $e=pq$ is abundant with respect to $(g,\mu)$. Let
    $\mu^*$ be an essentially optimal dual to the instance $\instance/
    e$. Let us define $\hat{\mu} \in \R_{++}^V$ by $\hat{\mu}_i :=
    \mu^*_i$ for $i \in V\setminus\{q\}$, and $\hat{\mu}_{\modified{p}} := \mu^*_{q} / \gamma_e$.
Then $\hat \mu$ is an essentially optimal dual to the instance $\instance$.
\end{lemma}
Consequently, if we can identify an abundant arc, we will be able to usefully recurse on the contracted instance.

\begin{algorithm}
    \caption{\RecursiveGenFlow}\label{alg:overall}

    \raggedright

    \begin{algorithmic}[1]
        \Require{Instance $\instance = (V,E,t,\gamma,b)$, fitting pair $(f,\mu)$ satisfying \nameref{cond:integral}.}
    \Ensure{An essentially optimal dual solution.} 
        
        \State $(f,\mu)\gets $\Call{Produce-Abundant-Arc}{$f,\mu$}.
        \If{$(f,\mu)$ is essentially optimal}
            \Return $\mu$.
        \EndIf
        \State Let $e=pq$ be an abundant arc with respect to $(f,\mu)$.
        \State $\mu^* \gets
        $\Call{\RecursiveGenFlow}{$\mathcal{I}/e, f/e, \mu/e$}.

        \State $\hat{\mu}_i \gets \mu^*_i$ for $i \in V \setminus \{q\}$; $\hat{\mu}_q \gets \mu^*_q / \gamma_e$.
         \State \Return $\hat\mu$.
    \end{algorithmic}
\end{algorithm}

\RecursiveGenFlow{} is described in Algorithm~\ref{alg:overall}.
The main work of the algorithm takes place in the subroutine \textsc{Produce-Abundant-Arc} (Section~\ref{sec:plentiful}).
This routine \edited{updates the flow and the labels}, and terminates with a fitting pair that is either essentially optimal, or contains an abundant arc.
In the former case we are done.
In the latter case, the abundant arc is contracted, and the algorithm recursively called on the contracted instance; the contracted images of the current primal-dual pair are used as initial solutions in the recursive call. 
Finally, the essentially optimal dual to the contracted instance is pulled back, via Lemma~\ref{lem:pullback}.

\begin{proof}[Proof of Lemma~\ref{lem:pullback}.]
Let $g$ be the feasible flow fitting $\mu$ with $\Ex(g,\mu)\le
Ex(f,\mu)$ and $\|f-g\|_\infty\le \Deficit(f,\mu)$ as guaranteed by Lemma~\ref{lem:safe-flow}. Note that if
$e=pq$ is an abundant arc with respect to $(f,\mu)$, then it is also
abundant with respect to $(g,\mu)$, namely, $g^\mu_e\ge \Ex(g,\mu)$.

Let $g' = g/e$ and $\mu' = \mu/e$.     Let $V'$, $t'$ and $b'$ be the node set, sink and demands of $\instance/e$ respectively.
 Let $(f^*, \mu^*)$ be an essentially optimal fitting pair for $\instance' = \instance / e$,
    where in addition, $\|f^*-g'\|_1$ is minimal. 
    We first need the following claim:
\begin{claim*}
   The support of $f^{*} - g'$ does not contain any directed cycles.
\end{claim*}
\begin{proof}
    Consider $h = f^* - g'$.
    For a contradiction, suppose that $C$ is a cycle in $\supp(f^{*} - g')$. 
    Note that since $C \subseteq E_{g'}$,  $\gamma(C) = \gamma^{\mu'}(C) \leq 1$.
    Similarly $\ole{C} \subseteq E_{f^*}$, and so $\gamma(\ole{C}) = \gamma^{\mu^*}(\ole{C}) \leq 1$.
    Thus, $\gamma(C) = 1$.
    We can therefore augment $f^*$ by sending a positive amount of flow along $\ole{C}$ without
    changing $\nabla f^*_i$ at any node. This yields another optimal
    solution $f^{**}$ such that  $\|f^{**}-g'\|_1<\|f^{*}-g'\|_1$;
    note that $(f^{**},\mu^*)$ also satisfies essential
    optimality. 
This gives a contradiction.
\end{proof}

Let us define $\hat{f}$ to be a flow on $\obe{E}$ such that
$\hat{f} / e = f^*$, $\nabla \hat{f}_p = b_p$ (unless $p=t$, in which case we instead require $\nabla \hat{f}_q = b_q$), and either
     $\hat{f}_{pq}$ or $\hat{f}_{qp}$ is zero.
%
%
The proof will be completed by showing that $(\hat f,\hat \mu)$ is an essentially optimal fitting pair.

     Since $(f^*, \mu^*)$ is essentially optimal, $\nabla f^*_i = b'_i$ for all $i \in V' \setminus \{t'\}$ that can reach $t'$ in $E_{f^*}$.
     By the choice of $\hat{f}$ as well as Lemma~\ref{lem:contractadditive}, it follows that $\nabla \hat{f}_i = b_i$ for all $i \in V \setminus \{t\}$ that can reach $t$ in $E_{\hat{f}}$. 
    We also clearly have that $\hat{f}$ is supported on tight arcs of $\hat{\mu}$.
    So the only potential obstruction to the essential optimality of $(\hat{f}, \hat{\mu})$ is the possibility that $\hat{f}_{qp} > 0$ (since this corresponds to negative flow on $pq$).
    
    Let $\hat{h} := \hat{f} - g$. 
    According to the Claim above, $\hat{h}/e = f^* - g'$ does not
    contain any directed cycles, and the same must hold for $\hat{h}$.
    It follows that we can decompose $\hat{h}$ as $\hat{h} =
    \sum_{\ell=1}^k h^\ell$, where each $h^\ell$ is a path flow from
    some node $p^\ell$ where $\nabla \hat{f}_{p^\ell} < \nabla
    g_{p^{\ell}}$ to some node $q^\ell$ where $\nabla \hat{f}_{q^\ell}
    > \nabla g_{q^{\ell}}$.
    Such a decomposition is easy to construct by using a topological ordering of $\supp(\hat{h})$. 
We note that it is also a special case of the general flow
decomposition argument in, e.g., \cite{Goldberg91}; since $\supp(h)$ is acyclic, four out
of the five types of elementary flows, Types II-V, cannot
exist.)

    Let $P^\ell$ denote the support of $h^\ell$, and $\lambda^\ell$ the value of $h^\ell$ on the first arc of $P^\ell$. 
    Since $\mu$ fits $g$, and $P^\ell\subseteq \supp(\hat{h}) \subseteq E_{g}$, we have $\gamma_{ij}^\mu \leq 1$ for all arcs of $P^\ell$, and therefore the relabelled flow $({h^\ell})^\mu$ is nonincreasing along $P^\ell$. 
    By optimality of $f^*$, $\nabla g'_{t'} \leq \nabla f^*_{t'}$.
    We claim that this implies that $\nabla g_t \leq \nabla \hat{f}_t$.
    This is immediate if $t \notin \{p,q\}$.
    If $t=q$, then $\nabla g_p \geq b_p = \nabla \hat{f}_p$, and so by Lemma~\ref{lem:contractadditive} $\nabla g_q \leq \nabla \hat{f}_q$.
    The case $t=p$ is similar.
    
    Thus, for any arc $a \in \obe{E}$,
\[
    \hat{h}^\mu_{a}\leq \sum_{\ell=1}^k ({h^\ell_{a}})^\mu
    \leq  \sum_{\ell=1}^k \frac{\lambda^\ell}{\mu_{p^\ell}}
    = \sum_{i\in V} \max\{\nabla (g - \hat{f})^\mu_{i},0\}
    \leq \sum_{i\in \Vm} \max\{\nabla g^\mu_{i}-b^\mu_i,0\}
    =\Ex(g,\mu).
\]
Since $g^\mu_{e} \geq \Ex(g, \mu)$ by assumption, and $\hat{f} = g + \hat{h}$, it follows that $\hat{f}_{qp} = 0$, as required.
\end{proof}

\subsection{Obtaining an abundant arc}\label{sec:plentiful}
A key step is a call to the subroutine \textsc{Highest-gain}$(\hat V,\hat E,\hat \gamma,Q')$,
where the input is a directed graph $(V',E')$ with gain factors
$\gamma'\in [0,1]^{E'}$, and $Q'\subseteq V'$. This returns the
labels $\sigma\in \R_{+}^{V'}$ where
\[
\sigma_i=\max\{\gamma'(P): P \mbox{ is a directed walk from $i$ to
  $Q'$}\}.
\]
If $i$ cannot reach $Q'$ in $E'$, then we define $\sigma_i=0$.
Note that $\sigma_i\le 1$ for all $i\in V'$.
Note that computing a highest gain augmenting path with respect to
the gains $\gamma'_e$ is equivalent to computing a shortest path for
the nonnegative weights $-\log \gamma'_e$. 
\modified{A standard way to avoid computations with logarithms by directly implementing 
\textsc{Highest-gain} as a multiplicative variant of Dijkstra's
algorithm. The running time is the same $O(m+n\log n)$ as for the standard (additive) Dijkstra algorithm. The same running time can be obtained if the graph is the extended network of $m+n$ nodes and $2m$ arcs obtained by reducing from the capacitated version~\cite{Orlin93}.}

A more technical subroutine is \textsc{Round-Labels}$(f,\mu)$.
This makes small perturbations to the labels (while preserving $f^{\mu}$ and all important structure) in order to ensure that the encoding lengths of labels remain small during the algorithm.
We discuss the implementation of this subroutine, and the encoding length bounds, in Section~\ref{sec:numerics}.
For now, we require only that the subroutine satisfies the following properties: if $(f', \mu') = \textsc{Round-Labels}(f,\mu)$, then

\medskip

\begin{compactenum}[(R1)]
    \item \label{item:rmaintain} $f^{\prime\mu'} = f^\mu$;
    \item \label{item:rtight} $(f', \mu')$ is a fitting pair with $E^{\mu'} \supseteq E^{\mu}$; and
    \item \label{item:rbounds} 
        for all $i \in \Vm$,
        $|b^\mu_i| \leq |b^{\mu'}_i| \leq \lceil |b^\mu_i| \rceil$.
\end{compactenum}

\medskip

\begin{algorithm}
    \caption{\textsc{Produce-Abundant-Arc}}\label{alg:abundant}
    \begin{algorithmic}[1]
        \Require{Fitting pair $(f,\mu)$ satisfying \nameref{cond:integral}.}
        \Ensure{Fitting pair $(f,\mu)$ satisfying \nameref{cond:integral} which is either essentially optimal, or contains an abundant arc.}

        \While{$f^\mu_e < \Ex(f,\mu) + \Deficit(f,\mu)$ for all $e \in E$}\Comment{\emph{Stop if there is an abundant arc}}
            \LineComment{Augmentation part of the iteration}
            \State $Q \gets \{ i \in \Vm : \netflow{f^\mu}{i}<b_i^\mu \}$.  \label{alg:Qdef}
            \State $R \gets \{ i \in \Vm : \netflow{f^\mu}{i} \geq b^\mu_i + 1\}$.
            \While{there is a tight path in $E_f$ from some $i \in R \cup \{t\}$ to some $j \in Q \cup \{t\}$, with $i \neq j$}
            \State Augment $f^\mu$ by sending $1$ unit from $i$ to $j$ along tight arcs (maintaining acyclicity of $f$). \label{alg:aug} 
            \State \textbf{if} there exists an abundant arc \textbf{then} \Return $(f,\mu)$.
	    \State Update $Q$ and $R$.
        \EndWhile 
        \vspace{.3cm}
        \LineComment{Label update part of the iteration}\label{alg:rescale}
        \State $\Tc \gets \{ i \in V: \exists \text{ a tight
          $i$-$t$ path and a tight $t$-$i$ path in } E_f
        \}$. \label{step:T-def}
        \State \label{step:aux} $\Vaux \gets V \cup \{\aux\}$; $\Eaux \gets E_f \cup \{ \aux i : i \in R \cup \VN \cup (\VP \cap T) \}$.
        \State \label{step:gammaaux} $\gammaaux_e \gets \gamma^\mu_e$ for all $e \in E_f$; $\gammaaux_{\aux i} \gets 1$ for $i \in R$;
        $\gammaaux_{\aux i} \gets -b^\mu_i / (1 - \nabla f^\mu_i)$ for $i \in \VN \setminus R$; and
        $\gammaaux_{\aux i} \gets b^\mu_i / (1 + \nabla f^\mu_i)$ for $i \in \VP \cap T$.
        \State $\sigma\gets $\Call{Highest-gain}{$\Vaux,\Eaux,\gammaaux,Q\cup\{t\}$}.

       \If{$\sigmab>0$} \Comment{\emph{Update labels}}
            \State $S\gets \{i\in V: \sigma_i\ge \sigmab \}$. \label{alg:S1}
            \State $\mu_i\gets \mu_i\sigmab/\sigma_i$ for $i \in S$.
            \State $f_{ij}\gets f_{ij} \sigmab/\sigma_i$ for $ij \in E[S]$. \Comment{\emph{Keep $f^\mu$ unchanged.}} \label{alg:fu}
            \State\Call{Round-Labels}{$f,\mu$}.
       \Else \Comment{\emph{Produce an essentially optimal solution}}
        \State $S\gets \{i\in V: \sigma_i>0\}$. \label{alg:S2}
       \State $f_{ij}\gets 0$ for $ij \in E[S]$. \label{alg:finalf}
          \State \Return $(f,\mu)$. \label{alg:essopt} 
        \EndIf
    \EndWhile
    \State \Return{$(f, \mu)$.} 
    \end{algorithmic}
\end{algorithm}

Algorithm~\ref{alg:abundant} gives the description of \textsc{Produce-Abundant-Arc}$(f,\mu)$. 
The output of the algorithm is either an abundant arc, or an
essentially optimal solution. 
Each iteration consists of a primal update part {\em (path augmentation)} followed
by a dual update part {\em (label update)}.
By augmenting $f^\mu$ by $1$ on a tight path
$P\subseteq E_f$, we mean increasing $f_{ij}$ by $\mu_i$ for
every original arc $ij$ in $P$ and reducing $f_{ji}$ by $\mu_i$ on every reverse arc $ij\in P$. We let $R$ denote the set of nodes with
relabelled excess at least 1, and $Q$ denote the set of nodes with negative excess.
The primal update comprises two types of path augmentations:
{\em (i)} the first type sends one  unit of relabelled flow
from a node in $R$  to the sink or to a node in
$Q$; {\em (ii)} the second type sends one unit of
relabelled flow from the sink to a node in $Q$.
To maintain the acyclicity of $f$ when augmenting from $i$ to $j$, we simply choose a path with as few arcs outside $\supp(f)$ as possible.

We continue to the label updates if no more path augmentations are possible. The label updates will guarantee the possibility of a
path augmentation in the next iteration, by creating new tight
paths. An important property of the label updates is that the relabelled flow $f^\mu$ does not change---and thus, the flow value $f_{ij}$ is updated whenever $\mu_i/\mu_j$ is modified, see line~\ref{alg:fu}. 
We note that it suffices to store only the relabelled flow $f^\mu$ during the algorithm, rather than $f$ itself.
This will remain conveniently integral: it is unchanged in the label updates and can only change by $\pm 1$ during path augmentations.

To obtain the label updates, we run the subroutine 
\textsc{Highest-Gain}$(\Vaux,\Eaux,\gammaaux,Q\cup\{t\})$, where $(\Vaux,\Eaux)$ is an auxiliary graph obtained from $(V,E_f)$ by adding a new node $\aux$ to $V$, with new edges $\aux i$ for certain nodes $i$ (lines~\ref{step:aux}--\ref{step:gammaaux}).
To understand the role of these edges and their gain factors $\gammaaux_{\aux i}$, suppose we run 
\textsc{Highest-Gain}$(V,E_f,\gamma,Q\cup\{t\})$ instead to obtain labels $\sigma_i$.
Let $S$ be the set of nodes that can reach $Q\cup \{t\}$ in $E_f$, and let $\sigma_\aux$ be the smallest label on $S$; note that $\sigma_i=0$ for $i\notin S$.
Then, with respect to the new labels $\mu'_i=\mu_i\sigma_\aux/\sigma_i\le \mu_i$, every node $i\in S$ can reach $Q\cup \{t\}$ in $E_f$ on  a tight path. However, these label changes will affect the excesses and deficiencies of some nodes. Since label changes do not change the relabelled flow values, we have  $\nabla f^{\prime \mu'}_{i}=\netflow{{f}^\mu}{i}$. On the other hand, the values $b^\mu_i$ may increase if $b_i>0$ and decrease if $b_i<0$. Consequently, the excess of nodes in $\VP$ may go down and the excess of nodes in $\VN$ may go up. \nnote{Not sure about my changes in this paragraph, or if I understood exactly what was there. So in this description \emph{all} labels are updated. The role of updating only nodes in $S$ doesn't seem to make an appearance in this description.}\lnote{Fixed. please take a look if this makes sense, and if it helps to clarify the role of $\aux$ or just adds to the confusion.}

The role of the new edges $\aux i$ is to keep all such changes under control.
These edges will guarantee that we stop with the label updates at the first instance one of the following events happen: \emph{(a)} a node in $R$ gets a tight path to $Q\cup\{t\}$; \emph{(b)} a node in $\VN\setminus R$ enters $R$ due to an increase in its excess; or \emph{(c)} a node in $\VP$ enters $Q$ due to a decrease in its excess.

In all of the above cases, $\sigma_\aux>0$, that is, $\aux$ can reach $Q\cup \{t\}$ in $E_f$. 
It is also possible that $\sigma_\aux=0$, in which case we will show that for all nodes $S$ with $\sigma_i>0$, we have $S\setminus \{t\}\subseteq \VZ$. We can set the flow  inside $S$ to $0$ and return  $(f,\mu)$ as an essentially optimal pair (lines~\ref{alg:S2}--\ref{alg:essopt}).

The set $\Tc$ is defined as the set of nodes that have  tight paths to and from $t$ in $E_f$ (line~\ref{step:T-def}). In case {\em (c)} above, we will only consider nodes in $\VP\cap \Tc$, i.e., we let the excess of nodes $i\in \VP\setminus \Tc$ decrease below $b_i^\mu-1$. The reason for this choice is somewhat subtle: if we stop the label updates because the excess of a new node  $i\in\VP$ becomes negative---thus, $i$ enters $Q$---then we would like to send flow from $t$ to $i$ in the next iteration. A tight $t$-$i$ path will however only exist if  we had $i\in T$ to begin with. It turns out that even with this choice, we can argue that the total deficiency of the nodes remains bounded (Lemma~\ref{lem:excessbounds}).

\section{Analysis}\label{sec:analysis}
In this section, we prove the following result.
\begin{theorem}\label{thm:alg-1}   
Algorithm~\ref{alg:roundmaxflow} is strongly polynomial and terminates with an
    essentially optimal dual
    solution in $O(nm(m+n\log n)\log(n^2/m)$ arithmetic operations. 
\end{theorem}
The proof of  Theorem~\ref{thm:main} will be completed in
Section~\ref{sec:first-phase}, where we show that the same
 running time guarantee can be obtained even without  assuming
 condition \nameref{cond:init} on the availability of an initial
 fitting pair of primal and dual feasible
 solutions. 

We define an \emph{augmentation} to be either a path augmentation, 
as performed in line \ref{alg:aug} of \textsc{Produce-Abundant-Arc}, 
or what we will call a \emph{null augmentation}: an event when a node $i\in \VN$ satisfies $\netflow{f^\mu}{i}<b_i^\mu$ at the start of a label update, but $\netflow{f^\mu}{i}\ge b_i^\mu$ after. 
Unlike path augmentations, these do not modify the solution at all; they are defined purely for accounting purposes.
As revealed in the analysis, null augmentations share some important
features with path augmentations.

The proof of Theorem~\ref{thm:alg-1} is given in four parts. In
Section~\ref{sec:correct}, we prove correctness: if the algorithm
terminates, it returns an essentially dual optimal solution.
Sections~\ref{sec:run-1} and \ref{sec:refined} provide the
 bound on the total number of 
 augmentations. Section~\ref{sec:numerics} presents the
 subroutine \textsc{Round-Labels}, and shows that the size of the numbers in the calculations remains
polynomially bounded in the input size (in other words, the algorithm runs in PSPACE).

Throughout this section, the values of $n$ and $m$ will always refer to the number of nodes and arcs in the original input graph. 
Hence in later recursive calls of the algorithm, the graph will have less than $n$ nodes.

\subsection{Correctness}\label{sec:correct}
In this section we prove that if Algorithm~\ref{alg:roundmaxflow}
terminates, then it returns an essentially optimal dual solution.
We start by showing some basic properties of
\textsc{Produce-Abundant-Arc}.
\begin{lemma}\label{lem:plentiful-basic}
    The following properties hold in every label update step in \textsc{Produce-Abundant-Arc} 
where $\sigmab > 0$.
\begin{enumerate}[(i)]
    \setlength{\itemsep}{1pt}
    

    \item \label{item:sigma}
 The vector $\gammaaux$ is well-defined, and $\sigmab<1$.

    \item \label{item:fittingpair} 
        $(f, \mu)$ remains a fitting pair, and there is a tight
 path in $E_f$ from every $i\in \Sc$ to $Q\cup\{t\}$. 

    \item\label{item:funchanged}
        $f^\mu$ remains unchanged. 

    \item\label{item:muincreases} 
        For each $i \in \Vm$, $|b^\mu_i|$ is nondecreasing.

    \item\label{item:augment-terminate} 
     Either at least one augmentation will be performed in the next iteration, or a null augmentation occurred during the label update.

    \item \label{item:bbounds} 
        After the label update,
        $\netflow{f^\mu}{i} \geq  b_i^\mu-1$ 
        for $i\in T$, and 
        $\netflow{f^\mu}{i} \leq b_i^\mu + 1$
        for $i \in \Sc  \setminus R$. Further, the label of every
        $i\in R$ remains unchanged until the call to \textsc{Round-Labels}.
\end{enumerate}
\end{lemma}
\begin{proof}
    Let $(f,\mu)$ denote the fitting pair at the start of the label update, and $(f', \mu')$ the values after, excluding the call to \textsc{Round-Labels}.
    It is immediate from properties \rpref{item:rmaintain}--\rpref{item:rbounds} that \textsc{Round-Labels} does not invalidate any of the claims.
We will repeatedly use the next property that follows by the definition of $\sigma_i$.
\begin{equation}\label{eq:sigma}
\sigma_i \geq \sigma_j\gamma^{\mu}_{ij}\quad \forall ij \in E_f.
\end{equation}
\begin{enumerate}[(i)]
    \item
        We need to show that the denominators in the definitions of $\gaux{i}$
        are all nonzero. We also show that $\gaux{i} \le 1$, with equality only if
$i\in R$. Assume $i\in \VN\setminus R$. Then,
$\netflow{f^\mu}{i}<b_i^\mu+1$, and thus $0<-b_i^\mu <
1-\netflow{f^\mu}{i}$. This also shows $\gaux{i}<1$. Next, let $i\in
\VP\cap T$. Then, $0<b_i^\mu\le \netflow{f^\mu}{i}$ as otherwise a path
augmentation from $t$ to $i$ would be possible. Again, we see that $\gaux{i}<1$.

Next, we show that $\sigmab<1$. Indeed, $\sigmab=1$
is only possible if $\sigma_i=\gaux{i} =1$ for some $i\in \Vm$. This in
particular yields $i\in R$, and that $i$ is connected to $Q\cup\{t\}$
by a tight path in $E_f$. Thus, it would have been possible to augment
on a path starting from $i$, a contradiction.

    \item Note that $E_{f'}=E_f$ since $f'_{ij}>0$ if and only if $f_{ij}>0$.
Thus, we need to check that $\gamma^{\mu'}_{ij} \leq 1$ for all $ij
\in E_f$. If $i$ and $j$ are both inside $\Sc$, then 
        $\gamma^{\mu'}_{ij} = \gamma^\mu_{ij} \sigma_j / \sigma_i \leq
        1$ by \eqref{eq:sigma}.    If $i$ and $j$ are both outside
        $\Sc$, then $\gamma^\mu_{ij}$ is unchanged. If $ij\in
        \delta^-(S)$, then $\sigma_j \geq \sigmab > \sigma_i$ and hence 
        $\gamma^{\mu'}_{ij}  = \gamma^{\mu}_{ij} \cdot \sigma_j /
        \sigmab < \gamma^{\mu}_{ij} \sigma_j / \sigma_i \leq 1$. 
Finally, if $ij\in\delta^+(S)$, then 
$\gamma^{\mu'}_{ij} = \gamma^\mu_{ij} \cdot \sigmab / \sigma_i \leq
\gamma^{\mu}_{ij} \leq 1$ since $\sigma_i\ge \sigmab$.

The second part of the claim follows by the definition of
\textsc{Highest-Gain}: every node $v\in V$ with $\sigma_v>0$ has a
path $P$ in $E_{f'}=E_f$ to $Q\cup\{t\}$ such that every edge $ij$ in
$P_v$ satisfies \eqref{eq:sigma} at equality. For every $i\in S$, the
entire path $P_v$ remains inside $S$, and after the relabelling, $P_v$
becomes a tight path.

    \item
        Consider any arc $ij \in E$.
        If $i,j \notin \Sc$, then $f'_{ij} = f_{ij}$ and $\mu'_i = \mu_i$. 
        if $i,j \in \Sc$, then $f'_{ij} / f_{ij} = \mu'_i / \mu_i = \sigmab/\sigma_i$. 
        Arcs $ij\in \delta^-(\Sc)$ cannot be tight, as otherwise
        \eqref{eq:sigma} would contradict $\sigma_i < \sigma_j$.
        This further implies that $f_{ij} = 0$ for every $ij\in \delta^+(\Sc)$, 
        since otherwise $ji\in E_f$ would be a tight arc entering $\Sc$.
        In all cases, $f^{\prime \mu'}_{ij} = f^\mu_{ij}$.

    \item This is immediate, since $\mu'_i = \mu_i$ for $i \notin \Sc$, and $\mu'_i = \mu_i \sigmab/\sigma_i$ for $i \in \Sc$, 
        where $\sigma_i \geq \sigmab$ by the definition of $S$.

    \item 
        Consider some $v \in R \cup \VN \cup (\VP \cap T)$ so that $\sigmab = \sigma_v\gaux{v}$.
        We have 
        \[ \nabla f^{\prime \mu'}_v - b^{\mu'}_v = \nabla f^\mu_v - b^\mu_v / \gaux{v}. \]
        Take the tight $v$-$(Q \cup \{t\})$ path $P$ with respect to $\mu'$, as
        guaranteed by part \eqref{item:fittingpair}.
        If $v \in R$, then $\gaux{v} = 1$ and $\nabla f^\mu_v - b^\mu_v \geq 1$, 
        and if $v \in \VN \setminus R$, $\nabla f^\mu_v - b^\mu_v / \gaux{v} = 1$ from the definition of $\gaux{v}$.
        In either case, if $v\notin Q$, then $P$ will be a possible augmenting path in the next iteration.
      If $v\in Q \cap (\VN \setminus R)$, then a null augmentation has occurred.
  
        Finally, if $v \in \VP \cap T$, then the definition of $T$
        guarantees the existence of a closed walk $C$ of tight edges
        in $E^{\mu'}_f$ containing both
        $t$ and $v$. Since $\gamma^{\mu'}(C)=\gamma^\mu(C)=1$, and
        $\gamma^{\mu'}_e\le 1$ for all $e\in E_{f'}=E_f$, it follows
        that every edge in $C$ is tight also with respect to
        $\mu'$. Consequently, $C$ contains a tight $t$-$v$ path $P'$ with
        respect to $\mu'$, and $\nabla f^\mu_v - b^{\mu}_v / \gaux{v} = -1$, again from
        the definition of $\gaux{v}$. Hence, $P'$ is a possible augmenting path.

    \item 
        Consider any $i \in T$. Then, $ \nabla f^\mu_i\ge  b^\mu_i$
        before the augmentation. If $i\notin \VP$, then
        $b^{\mu'}_i\le b^\mu_i$ (part \eqref{item:muincreases}). If
        $i\in \VP\cap T$, then 
        $\sigmab \geq \sigma_i \gaux{i}$, and so $\sigmab / \sigma_i \geq b^\mu_i / (1 + \nabla f^\mu_i)$.
        Hence 
        \[ \nabla f^{\prime \mu'}_i - b^{\mu'}_i = \nabla f^\mu_i - b^\mu_i \sigma_i / \sigmab \geq -1. \]
        
        Now consider any $i \in \Sc \setminus R$. Then,  $ \nabla
        f^\mu_i<  b^\mu_i+1$ before the augmentation, and 
if $i \notin \VN$, then $b^{\mu'}_i\ge b^\mu_i$. In case $i\in \VN\setminus R$,
then $\sigmab \geq \sigma_i \gaux{i}$, so $\sigmab / \sigma_i \geq -b^\mu_i / (1 - \nabla f^\mu_i)$.
        Hence 
        \[ \nabla f^{\prime \mu'}_i - b^{\mu'}_i = \nabla f^\mu_i -
        b^\mu_i \sigma_i / \sigmab \leq 1. \]
Finally, let $i\in R$. We clearly have $\mu'_i=\mu_i$ if  $i\in
V\setminus S$. If $i\in S$, then we must have
$\sigmab=\sigma_i\gaux{i}=\sigma_i$. Hence, $\mu_i'=\mu_i\sigmab/\sigma_i=\mu_i$.  
\end{enumerate}
\end{proof}
Next we show that despite not making any effort to maintain the feasibility of $f$, 
safety is preserved.

\begin{lemma}\label{lem:safe}
The labelling $\mu$ remains safe throughout Algorithm~\ref{alg:overall}.
\end{lemma}
\begin{proof}
Safety of the input $\mu$ is required, and safety is obviously preserved by the contraction of a tight arc.
The nontrivial part is to show that it is also maintained during the label update steps in \textsc{Produce-Abundant-Arc}. 
So assume $\mu$ is safe before a label update, and let $\mu'$ denote
the updated labels (ignoring the call to \textsc{Round-Labels}; since this routine preserves tight arcs, it clearly preserves safety). Let $S_0\subseteq V$ denote the set that can reach
$Q\cup\{t\}$ on a tight path in $E_f$ before the label update. 
Clearly, $\sigma_i=1$ for all $i\in S_0$.
Since $\mu$ is safe, there is a flow $g$ certifying this fact: $(g, \mu)$ is a fitting pair, and $g$ is feasible.
By the definition of $S_0$, there are no tight arc with respect to $\mu$ in $\delta^-(S_0)$, and $g$ is supported only on tight arcs, $g(\delta^-(S_0)) = 0$.
We now construct a flow $g'$ certifying the feasibility of $\mu'$ as follows:
\[ g'_e = \begin{cases} g_e & \text{for } e \in E[S_0]\\
        f_e & \text{otherwise.}
    \end{cases}
\]
Since $f$ does not send any flow either into or out of $S_0$, $g'_e =
0$ for all $e \in \delta^+(S_0) \cup \delta^-(S_0)$.
Note that for every $ij\in \supp(f)$, we have $\sigma_i=\sigma_j$, and
therefore $\gamma^{\mu'}_{ij}=1$. We therefore see that $(g',\mu')$ is
also a fitting pair.
Now for any $i \in S_0$, 
\[ 
    \nabla g'_i = \nabla g_i + \sum_{j:ij \in \delta^+(S_0)}g_{ij}  \geq b_i\,,
\]
whereas for any $i \notin S_0$, 
\[ \nabla g'_i = \nabla f_i \geq b_i\, , \]
since $Q \subseteq S_0$.
So $g'$ is feasible, and hence $\mu'$ is indeed safe.
\end{proof}
In the analysis of the case $\sigmab=0$ as well as later on we will
use the following simple claim:
\begin{claim}\label{cl:S0}
    Let $S_0$ denote the set of nodes that can reach $Q \cup \{t\}$ along a tight path in $E_f$.
Then, there are no
    tight arcs in $E$ from $T$ to $S_0 \setminus T$. 
\end{claim}
\begin{proof}
Assume for
    a contradiction that there is a
    tight arc $ij\in E^\mu_f$ with $i\in T$ and $j\in S_0\setminus T$. By
    the definition of $T$, there is a tight path $P$ from $t$ to $i$,
    and by the definition of $S_0$, there is a tight path $P'$ from
    $j$ to some $q\in Q\cup\{t\}$. If $q\in Q$, then $P\cup P'$ contains a tight
    $t$-$q$ walk $P''$, a contradiction since it would have been
    possible to augment on $P''$. On the other hand, $q=t$ implies
    that $j\in T$, because $P'$ is a tight $j$-$t$ path, and appending
    $ij$ to $P$ yields a tight $t$-$j$ path. This 
    contradicts  $j\in S_0\setminus T$.
\end{proof}

\begin{lemma}\label{lem:essopt-termination}
    If $\sigmab=0$ in \textsc{Produce-Abundant-Arc}, then it terminates with an essentially optimal fitting pair.
\end{lemma}
\begin{proof}
  In this case, $S$ is precisely the set of nodes that can reach $Q\cup\{t\}$ in $E_f$.
    Let $f$ denote its value before the final change in line~\ref{alg:finalf}, and $f'$ its final value.
    Thus $\delta^-_{E_f}(\Sc) = \emptyset$, i.e., $\delta^-(\Sc) = \emptyset$ and $f_{ij} = 0$ for all $ij \in \delta^+(\Sc)$.

    To prove essential optimality of $(f', \mu)$, it suffices to show
    that $\Sc \subseteq V^0$.
    For then $\nabla f'_i = 0 = b_i$ for all $i \in \Sc \setminus \{t\}$ (and hence all $i$ that can reach $t$ in $E_{f'}$, since $\delta^-(\Sc) = \emptyset$); 
    and $\nabla f'_i = \nabla f_i \geq b_i$ for all $i \notin \Sc$, since $Q \subseteq \Sc$.

    Since $\sigmab = 0$, there are no paths in $E_f$ from any node in $R \cup \VN \cup (\VP \cap T)$ to $Q \cup \{t\}$, and hence all these nodes lie outside of $\Sc$.
    It remains to show that all nodes of $\VP \setminus T$ also lie outside of $\Sc$.

We let $S_0$ as in Claim~\ref{cl:S0}.  Clearly, $T \subseteq S_0 \subseteq S$. 
Claim~\ref{cl:S0} asserts that there are not tight arcs from $T$ to
$S_0\setminus T$. 
    Since there are also no tight arcs in $\delta^-(S_0)$, we have that $\delta^-_{E^\mu}(S_0 \setminus T) = \emptyset$. 
    Since $\mu$ is safe, it follows by Lemma~\ref{lem:safe-cut} that 
    $b^\mu(S_0\setminus T)\leq 0$.

    Next, consider $S \setminus S_0$. We claim that $f(\delta^-(S
    \setminus S_0)) = 0$.
Clearly, there are no arcs $ij\in E$ with $i\notin S$ and $j\in
S\setminus S_0$. Assume $f_{ij}>0$ for some arc $ij\in E$ with $i\in S_0$,
$j\in S\setminus S_0$. Then $ji\in E^\mu_f$, and thus $j\in S_0$, a
contradiction.
Since $Q \subseteq S_0$, we see that
\[
   b^\mu(S\setminus S_0)\leq \sum_{i\in S\setminus
  S_0}\netflow{f^\mu}{i}=f(\delta^-(S\setminus
S_0))-f(\delta^+(S\setminus S_0))\leq 0.
\]

    Altogether, we conclude that $b^\mu(S\setminus T) \leq 0$. As we
    already know that $S\setminus T\subseteq \VP\cup V^0$, it follows that $\VP \cap (S \setminus T) = \emptyset$, as required.
\end{proof}

The correctness of Algorithm~\ref{alg:roundmaxflow} follows by the above
statements and by Lemma~\ref{lem:pullback}. 
%

\subsection{Bounding the number of augmentations}\label{sec:run-1}
%
In this section, we will set up the required potential analysis, 
and prove a strongly polynomial bound on the overall number of augmentations.
This analysis will be further improved in Section~\ref{sec:refined} to 
obtain the running time bound in Theorem~\ref{thm:alg-1}.

Throughout this section and the following one, let $\bar{\mathcal{I}} = (\bar{V}, \bar{E}, \bar{t}, \bar{\gamma}, \bar{b})$ denote the initial instance provided to Algorithm~\ref{alg:roundmaxflow}; so $n=|\bar{V}|$ and $m=|\bar{E}|$.
Let $(\fs,\mus)$ denote the primal-dual pair satisfying \nameref{cond:integral} and for which $-1 < \nabla f^\mu_i - b^\mu_i < 2$ for all $i \in \bar{V} \setminus \{\bar{t}\}$, supplied by Algorithm~\ref{alg:roundmaxflow} to the first invocation of \RecursiveGenFlow. 
We use $\instance=(V, E, t, \gamma, b)$ to denote the instance and 
 $(f,\mu)$ to denote the fitting pair at some point in the execution of the algorithm.

\begin{lemma}\label{lem:excessbounds}
    $\Ex(f,\mu) \leq 2n$ and $\Deficit(f,\mu) \leq 3n$ throughout
    Algorithm~\ref{alg:overall}.
\end{lemma}
The proof will require the following property of the set $T$:
\begin{lemma}\label{lem:T}
 After a label update step in 
\textsc{Produce-Abundant-Arc} 
with $\sigmab > 0$, there is no tight arc in $E$  leaving the set $T$.
\end{lemma}
\begin{proof}
Let $\mu$ denote the labels before and $\mu'$ after the label update.
For a contradiction, let $ij\in E^{\mu'}$ with $i\in T$ and $j\in
V\setminus T$. Since $i$ has a tight path to $t$ w.r.t. $\mu$, it
follows that $\sigma_i=1$. Recall that
$\gamma_{ij}^{\mu'}=\gamma_{ij}^\mu \sigma_j/\sigma_i\le \sigma_j$
using that $\gamma_{ij}^\mu\le 1$. Hence, $ij\in E^{\mu'}_f$ yields
$\sigma_j=1$ as well as $\gamma_{ij}^\mu= 1$. This contradicts
Claim~\ref{cl:S0}. Indeed, $\sigma_j=1$ is equivalent to $j\in S_0$,
and thus $ij$ would be a tight arc w.r.t. $\mu$ from $T$ to
$S_0\setminus T$.
\end{proof}

\begin{proof}[Proof of Lemma~\ref{lem:excessbounds}] Let us start with the bound $\Ex(f,\mu) \leq 2n$.
Define
\[
    \globpot(f,\mu) := \sum_{i \in \Vm} \max\{\lceil \netflow{f^\mu}{i}-b^\mu_i \rceil, 2\};
\]
clearly $\Ex(f,\mu) \leq \globpot(f,\mu)$, and since $\nabla \fs^{\mus}_i -b^{\mus}_i < 2$ for all $i \in \bar{V} \setminus \{\bar{t}\}$, $\globpot(\fs,\mus) = 2(n-1)$.
We prove that $\globpot(f,\mu)$ is nonincreasing throughout the algorithm, implying the bound on $\Ex(f,\mu)$.

Consider a call to
\textsc{Produce-Abundant-Arc}.
A path augmentation may increase $\netflow{f^\mu}{i}$ only for $t$ or a node
$i\in Q$, where $\netflow{f^\mu}{i}<b_i^\mu$. 
Hence, $\globpot(f,\mu)$ is not increased by path augmentations. 
We claim that no term in $\globpot(f,\mu)$ increases during label update steps. 
For $i\in V\setminus \Sc$, $\netflow{f^\mu}{i}- b_i^\mu$ is unchanged.
For $i\in \Sc \setminus R$, Lemma~\ref{lem:plentiful-basic}~(\ref{item:bbounds}) implies that
$\netflow{f^\mu}{i}- b_i^\mu\le 1$ after the label update.
For $i \in \Sc \cap R$, the same
Lemma~\ref{lem:plentiful-basic}~(\ref{item:bbounds}) implies that $\mu_i$ is unchanged by the label update until
the call to \textsc{Round-Labels}, which does not increase $\lceil \nabla f^\mu_i - b^\mu_i\rceil$ by property~\rpref{item:rbounds}.

We also have, using Lemma~\ref{lem:contractadditive}, that for any $e=pq$ tight for $\mu$,
\begin{align*}
    &\globpot(f,\mu) - \globpot(f/e,\mu/e) \\
                    &\quad= \max\{ \lceil \nabla f^\mu_p - b^\mu_p\rceil, 2\} + \max\{\lceil \nabla f^\mu_q - b^\mu_q \rceil, 2\} -
\max\{\lceil \nabla f^{\mu}_p - b^{\mu}_{p} + \nabla f^\mu_q - b^\mu_q \rceil, 2\}\\
&\quad\geq 0. 
\end{align*}
Thus, $\globpot(f,\mu)$ does not increase upon a recursive call to the algorithm.

\medskip
We now show $\Deficit(f,\mu)\le 3n$.
Initially, $\Deficit(\bar{f},\bar{\mu}) \leq n$.
Similarly to the above, $\Deficit(f,\mu)$ does not increase upon contracting a tight arc.
Clearly a path augmentation will not increase the deficit of any node in $\Vm$.

By Lemma~\ref{lem:plentiful-basic}~(\ref{item:bbounds})
and the requirements on \textsc{Round-Labels}, $b^\mu_i - \nabla
f^\mu_i \leq 1$ for all $i \in  T$.
Lemma~\ref{lem:T} asserts that $\delta_{E^\mu}^-(V \setminus T) =
\emptyset$.
The safety of $\mu$ implies that $b^\mu(V\setminus T)\leq 0$
(Lemma~\ref{lem:safe-cut}). We further claim that 
$f(\delta^+(V\setminus T))=0$. Indeed, if there were an arc
$ij\in\delta^+(T)$ with $f_{ij}>0$, then both $ij,ji$ would have been
tight edges in $E_f$ before the label update. This contradicts the
definition of $T$, since there would have been tight $j$-$t$ and
$t$-$j$ paths before the label update, both via $i$.
The above imply
\[ \sum_{i \in V\setminus T} \nabla f^\mu_i-b_i^\mu
=f(\delta^-(V\setminus T))-f(\delta^+(V\setminus T))-b^\mu(V\setminus T) \ge 0. \]
It follows that
\begin{align*} 
\Deficit(f,\mu)&=\sum_{i \in T} \max\{b^\mu_i - \nabla
f^\mu_i,0\}+\sum_{i \in V\setminus T} \max\{b^\mu_i - \nabla
f^\mu_i,0\}
\\
&\leq 
|T|+\sum_{i \in V \setminus T} \max\{\nabla f^\mu_i - b^\mu_i,0\} \\
&\leq |T|+\Ex(f,\mu)\le 3n.
\end{align*}
\end{proof}

We measure progress via the potential
\[
\apot(f,\mu) := \sum_{i \in \VP} \netflow{f^\mu}{i} \;-\; \sum_{i \in \VN} \netflow{f^\mu}{i}. 
\]

\begin{lemma}\label{lem:potbound}
    At any point in the algorithm where a path augmentation is performed, 
    $|\apot(f,\mu)| \leq 5 n^2$. 
\end{lemma}
\begin{proof}
    Since $\Ex(f,\mu) + \Deficit(f,\mu) \leq 5n$ throughout the algorithm, and $|\supp(f)| \leq n-1$, if $|\apot(f,\mu)| \geq 5(n-1)n$, then there must be some abundant arc.
    No path augmentation is performed if there are any abundant arcs, so the claim follows. 
\end{proof}
Let us now examine how $\apot(f,\mu)$ changes during iterations of \textsc{Produce-Abundant-Arc}.
Clearly only path augmentations have any effect, since $f^\mu$ is unchanged in the label update part.

Call a path augmentation that begins at a node in $\VN \cup \{t\}$ and ends at a node in $\VP \cup \{t\}$ a \emph{helpful} augmentation, 
and all other augmentations (including all null augmentations) \emph{unhelpful}.
Note that an unhelpful augmentation decreases $\apot(f,\mu)$ by at most $2$.
Any helpful augmentation, on the other hand, increases $\apot(f,\mu)$ by either $1$ or $2$.

\begin{lemma}\label{lem:unhelpful}
    There are at most $6n$ unhelpful augmentations in any call to \textsc{Produce-Abundant-Arc}.
\end{lemma}
\begin{proof}
    Let $(\tilde{f},\tilde{\mu})$ be the fitting pair at the start of the call.
    Consider any $i \in \VP$.
    Label update steps can only decrease the (relabelled) excess $\nabla f^\mu_i - b^\mu_i$ (by Lemma~\ref{lem:plentiful-basic}~(\ref{item:muincreases})),
    and a helpful augmentation will increase this excess, but only to a value below $1$.
    An unhelpful augmentation involving $i$, on the other hand, must decrease its excess, from a value of at least $1$.
    It follows that the number of unhelpful augmentations involving $i$ is not more than the excess $\max\{\nabla \tilde{f}^{\tilde{\mu}}_i - b^{\tilde{\mu}}_i,0\}$ of $i$ at the start of the call.
    Hence the total number of unhelpful augmentations starting from nodes in $\VP$ is at most $\Ex(\tilde{f},\tilde{\mu})$.

    Similarly, the total number of unhelpful augmentations ending at a node in $\VN$ is at most $\Deficit(\tilde{f}, \tilde{\mu})$.
    By Lemma~\ref{lem:excessbounds}, there are at most $5n$ unhelpful augmentations overall.
    
    Finally, the number of null augmentations is at most $n$.
    After a node in $\VN$ leaves $Q$ it will never re-enter, since label updates only increase its excess.
\end{proof}

This already gives a strongly polynomial bound on the number of augmentations.
Consider any sequence of path augmentations between contractions (or finding an essentially optimal fitting pair).
The value of $\apot(f,\mu)$ lies between $-5n^2$ and $5n^2$ throughout this sequence, and aside from at most $6n$ unhelpful augmentations that increase $\apot(f,\mu)$ by at most $12n$ in total, each augmentation decreases $\apot(f,\mu)$ by at least $1$.
So there can be at most $O(n^2)$ augmentations between contractions, and hence at most $O(n^3)$ augmentations in total.
In the case of dense graphs ($m=\Theta(n^2)$) this 
is sufficient to obtain the claimed running time of Theorem~\ref{thm:alg-1}.
For sparser graphs, a more refined analysis is needed.

\subsection{A refined bound on the number of augmentations}\label{sec:refined}
In this section, we prove the following more refined bound needed for Theorem~\ref{thm:alg-1}.
\begin{theorem}\label{thm:refined-aug}
    There are at most $O(mn \log (n^2/m))$ augmentations throughout the execution of Algorithm~\ref{alg:overall}.
\end{theorem}
%

In order to obtain this improved bound, we investigate how $\apot(f,\mu)$ changes through contractions.
In fact, we consider a different (but related) potential. 
Let
\[ \npot(f,\mu) := \sum_{i \in \Vm} |\nabla f^\mu_i|. \]
\begin{lemma}
    If the total decrease of $\npot$ due to contractions over the algorithm is $\Delta$, then the number of augmentations is $\Delta + O(n^2)$.
\end{lemma}
\begin{proof}
We first show that $\npot(f,\mu)$ and $\apot(f,\mu)$ are always close.
Consider the contribution of any node $i \in \Vm$ to $\npot(f,\mu)$ and to $\apot(f,\mu)$.
The contributions are equal if $\nabla f^\mu_i$ and $b^\mu_i$ are both positive or both negative, and otherwise, the contributions differ by at most $2|\nabla f^\mu_i|$.
If $\nabla f^\mu_i \geq 0$ and $b^\mu_i \leq 0$, then $|\nabla f^\mu_i| \leq |\nabla f^\mu_i - b^\mu_i|$, and similarly if $\nabla f^\mu_i \leq 0$ and $b^\mu_i \geq 0$.
Thus
\[ |\npot(f,\mu) - \apot(f,\mu)| \leq 2\sum_{i \in \Vm}|\nabla f^\mu_i - b^\mu_i| \leq 8n, \]
where the final inequality follows from Lemma~\ref{lem:excessbounds}.

The decrease in $\apot$ due to a contraction can thus be bounded by the decrease in $\npot$ plus $16n$.
Thus the total decrease in $\apot$ due to contractions is bounded by $\Delta + 16n^2$.

    Lemma~\ref{lem:potbound} shows that $|\apot(f,\mu)| \leq 5n^2$ at any point where an augmentation occurs.
    Since by Lemma~\ref{lem:unhelpful} there are $O(n)$ unhelpful augmentations per contraction, and hence $O(n^2)$ in total, the total decrease in $\apot$ due to unhelpful augmentations is at most $O(n^2)$.
    Since each helpful augmentation increases $\apot$ by at least one, the lemma follows.
\end{proof}

For the remainder of this section, it will be convenient to think of contractions as preserving the arc set;
thus when contracting an arc $e=pq$, the resulting graph may be non-simple, and arc $e$ will become a loop.
The sum of the indegrees (or outdegrees) of the graph at any moment will then always be precisely $m$. 

Consider some instance and fitting pair $(f,\mu)$ before contracting
an abundant arc $e$ with endpoints $p$ and $q$ (that is, $e=pq$ or $e=qp$).
By swapping the labels if necessary, we may assume that $p \neq t$.
Let $(f', \mu')$ be the resulting fitting pair after contraction.
If $t = q$, then 
\[ \npot(f,\mu) - \npot(f',\mu') = |\nabla f^\mu_p|. \]
If $t \neq q$, then
\begin{align*} 
    \npot(f,\mu) - \npot(f',\mu') &= |\nabla f^\mu_p| + |\nabla f^\mu_q| - |\nabla f^\mu_p + \nabla f^\mu_q| \\
                                &\leq 2\min\{|\nabla f^\mu_p|, |\nabla f^\mu_q|\}.
\end{align*}
Let $d_i = |\delta(i)|$ denote the total degree of $i \in V$ (a loop counts twice towards the degree).
Now since $f^\mu_a \leq 5n$ for all $a \in E$, and $f$ is acyclic, we deduce that $|\nabla f^\mu_i| \leq 5n\min\{d_i,n\}$ for all $i \in \Vm$.
Hence
\[ \npot(f, \mu) - \npot(f',\mu') \leq \begin{cases} 5n \min\{d_p,n\} & \text{ if $t=q$}\\
        10n\min\{d_p, d_q,n\} &\text{ if $t \neq q$}
    \end{cases}.
\]
We now charge this potential decrease according to the following scheme.
Each token has a value of $10n$.
If $t=q$, give a token to each arc in $\delta(p)$ (a loop at $p$ receives two tokens). 
If $t \neq q$, compare $d_p$ and $d_q$; assume by relabelling if necessary that $d_p \leq d_q$.
If $d_p \geq n$, we give $n$ tokens to a central store, 
and if $d_p < n$, then give a token to each arc in $\delta(p)$ (again, a loop receives two tokens).
The total value of tokens assigned is then at least the potential decrease.

\begin{lemma}
    The total number of tokens assigned over all contractions is at most $O(m\log(n^2/m))$, and hence the total decrease in $\npot$ is at most $O(mn\log(n^2/m))$.
\end{lemma}
\begin{proof}
    Consider the value $\kappa := \sum_{i \in \Vm} \min\{d_i,n\}$.
    Each time the central store receives $n$ tokens, $\kappa$ must decrease by $n$.
    Since $\kappa$ is throughout between $0$ and  $2m$ and is
    nonincreasing, the central store receives at most $2m$ tokens in
    total. 

    Consider an arc $a=ij \in \bar{E}$.
Each endpoint of $a$ will be contracted with the root at most once, and thus receives at most two tokens in this fashion.
Every other time that $a$ receives a token, one endpoint of $a$ doubles its degree (including loops).
Once an endpoint has degree $n$ or more, it does not cause $a$ to be charged further.
Thus $a$ receives at most $\lceil\log_2(n/d_i)\rceil + \lceil\log_2(n/d_j)\rceil + 2$ tokens ($d_i$ refers to the degree of node $i \in \bar{V}$).

Applying Jensen's inequality to the concave function $x\log_2(n/x)$, the total number of tokens assigned to arcs is at most
\[ \sum_{ij \in \bar{E}} (\log_2(n/d_i) + \log_2(n/d_j) + 4) =
        2\sum_{i \in \bar{V}} d_i\log_2(n/d_i) + 4m
        \leq 4m\log_2(n^2/2m) + 4m.
\]
\end{proof}
This completes the proof of Theorem~\ref{thm:refined-aug}.


\subsection{Bounding the work per augmentation}\label{sec:numerics}
In this section, we complete the proof of Theorem~\ref{thm:alg-1}.
First, we note that the algorithm requires $O(m + n\log n)$ arithmetic operations between augmentations.
As already observed, this is the number of operations needed for the call to \textsc{Highest-Gain}. 
Finding an augmenting path if one exists, as well as computing $T$,
can be done with breadth-first searches in time $O(m)$. Computing the
values of $\gaux{i}$ and updating the labels takes time $O(n)$.
The only other nontrivial step is \textsc{Round-Labels}, which we have not yet defined.
We will do so below---it will also require $O(m+n\log n)$
operations. The running time bound claimed in Theorem~\ref{thm:alg-1}
then follows by Theorem~\ref{thm:refined-aug}.

The final step to showing that our algorithm is strongly polynomial is 
to demonstrate, with an appropriately defined \textsc{Round-Labels} subroutine, that all numbers appearing during the algorithm 
have encoding lengths polynomially bounded in the input size. 
This was a nontrivial challenge in the previous strongly polynomial algorithm~\cite{laci}.
For our algorithm, we will see that this can be done rather cleanly. 

Consider an instance $\mathcal{I}=(V,E,t,\gamma,b)$, such that $\gamma\in \Q_{++}^{E}$ and $b\in \Q^{\Vm}$.
Let $B$ be an integer that exceeds the largest numerator or denominator appearing in the description of the $\gamma_e$ and $b_i$ values as
ratios of two integers.

We do not need to work with  the flow $f$ directly, but maintain the relabelled flow $f^\mu$ instead. 
This remains integral throughout, except at the very beginning and in the final computation of a primal solution.
Moreover, the values $f^\mu_e$ are strongly polynomially bounded.
Indeed, whenever $f^\mu_e>5n$ then $e$ is abundant according to Lemma~\ref{lem:excessbounds}.
Let us now turn to the labels $\mu$. 
\begin{definition}
Given a labelling $\mu \in \Q_{++}^V$, we say that a node $i \in \VP \cup \VN$ is an \emph{anchor} if $b_i / \mu_i$ is an integer, bounded by $5n^2$ in absolute value.
We also say that a node $i \in V^0 \cup \{t\}$ is an anchor if $\mu_i = 1$.
We say that a node $j \in V$ is \emph{anchored} if there is a path $P$ in $\obe{E^{\mu}}$ between $j$ and an anchor node.
\end{definition}
Note that if $i$ is an anchor, then $\mu_i$ can be written as a fraction with numerator and denominator both bounded by $5n^2B$.
If $j$ is anchored, then it can be written as $\gamma(P)\cdot \mu_i$ for some path $P \in \obe{E}$ and anchor $i$, and hence it may be written as a fraction with numerator and denominator both bounded by $5n^2B^n$.
We will construct \textsc{Round-Labels} in such a way that all nodes are anchored after the call, providing us with the required PSPACE guarantee.

The \textsc{Round-Labels} procedure is defined in Algorithm~\ref{alg:roundlabels}. 
It modifies the labels in two stages. 
In the first stage, we \emph{decrease} the label of certain nodes until they become anchored. 
The update is carried out via the subroutine \textsc{Highest-Gain}; every arc that was already tight remains tight during the update. 
Some nodes may become anchors due to the label decrease, 
whereas other nodes may become anchored as they become connected to an anchor node by a tight path.

During this first stage,  \textsc{Highest-Gain} obtains positive
multipliers $\theta_i>0$ for a subset $\hat V\subseteq V$; however, we
may get $\theta_i=0$ for a nonempty set $V\setminus \hat V$. This
means that nodes in $V\setminus \hat V$
cannot become anchored even if their labels are arbitrarily
decreased; in particular, every such node must be in $\VZ\cup\{t\}$,
and $\delta^+_E(\hat V)=\emptyset$. 
For these nodes, the first stage multiplies the labels by a factor $\theta^* \leq 1$ so that
$\gamma_{ij}^{\mu'}\le 1$ for all $ij\in \delta^-_E(\hat V)$, and also $\mu'_i \leq 1$ for all $i \in V \setminus \hat{V}$. 

The second stage of \textsc{Round-Labels} calls \textsc{Highest-Gain}
a second time, to \emph{increase} the labels on $V\setminus \hat V$, until they become anchored (labels in $\hat{V}$ are unaffected). 
Since $V \setminus \hat{V} \subseteq \VZ \cup \{t\}$, this does not cause difficulties with property \rpref{item:rbounds}.
Note that the second stage returns strictly positive multipliers $\kappa_i>0$ since a node $i\in
\VZ\cup\{t\}$ becomes an anchor once $\mu_i=1$.

\begin{algorithm}
    \caption{\sc Round-Labels}\label{alg:roundlabels}
    \raggedright
    \begin{algorithmic}[1]
        \LineComment{First stage}
        \State $(V', E') \gets (V \cup \{t'\}, \ole{E} \cup E^\mu \cup
            \{it' : i \in \VP\cup \VN \})$.
        \State   $\gamma'_{ij} \gets 1/\gamma^\mu_{ij}$ for all $ij\in \ole{E}\cup E^\mu$;
        $\gamma'_{it'} \gets |b^{\mu}_i|/ \lceil |b^\mu_i| \rceil$ for all $i \in \VP \cup \VN$.
        \State $\rndgain \gets \Call{\textsc{Highest-Gain}}{V', E', \gamma', \{t'\}}$.
        \State $\hat{V} \gets \{ i \in V : \rndgain_i > 0\}$; $\rndgain^*\gets \min\bigl\{\, 1,\, \min_{i \in V \setminus \hat{V}} \mu_i^{-1},\, \min_{ij\in \delta^-_E(\hat V)} \theta_j/\gamma_{ij}^\mu\,\bigr\}$.
        \State $\mu'_i \gets \mu_i \rndgain_i$ for all $i \in \hat{V}$;
        $\mu'_i \gets \mu_i \rndgain^*$ for all $i\in V\setminus \hat{V}$.

        \vskip.2cm

        \LineComment{Second stage}
        \State $(\bar V, \bar E) \gets (V \cup \{\bar t\}, E\cup \ole{E^{\mu}} \cup \{i\bar t: i\in V\setminus \hat V\}$). 
        \State  $\bar \gamma_{ij} \gets \gamma^{\mu'}_{ij}$ for all $ij\in E\cup \ole{E^{\mu}}$;
            $\bar \gamma_{i\bar t} \gets \mu'_i$ for all $i \in V\setminus \hat V$.
        \State  $\kappa \gets \Call{\textsc{Highest-Gain}}{\bar{V}, \bar{E}, \bar{\gamma}, \hat V\cup \{\bar{t}\}}$.
        \State  $\mu'_i \gets \mu_i'/\kappa_i$ for all $i\in V\setminus \hat{V}$. 

        \vskip.2cm
        
        \State $f'_{ij} \gets f_{ij} \cdot \mu'_i / \mu_i$ for all $ij \in E$. \label{alg:funchanged} 
        \State \Return{$(f',\mu')$}.
    \end{algorithmic}
\end{algorithm}

\begin{lemma}
    The output $(f', \mu')$ of \textsc{Round-Labels} satisfies all the required properties \rpref{item:rmaintain} through \rpref{item:rbounds}.
\end{lemma}
\begin{proof}
~
    \begin{enumerate}[\em(R1)]
        \item \emph{$f^{\prime\mu'} = f^\mu$}:
            This is immediate from the definition of $f'$ in line~\ref{alg:funchanged}.
    
        \item \emph{$(f', \mu')$ is a fitting pair with $E^{\mu'} \supseteq E^{\mu}$}:
            It suffices to show that $\gamma^{\mu'}_{ij} \leq 1$ for all $ij \in E \cup \ole{E^\mu}$.

            If $i,j \in \hat{V}$, then $ji \in E'$; by the definition of \textsc{Highest-Gain}, $\rndgain_j \geq \rndgain_i \gamma'_{ji}$.
            Recalling that $\gamma'_{ji} = 1/\gamma^{\mu}_{ji}=\gamma^{\mu}_{ij}$, we have $\gamma^{\mu'}_{ij} = \gamma^{\mu}_{ij} \rndgain_i / \rndgain_j \leq 1$.
            If $i,j \in V \setminus \hat{V}$, an essentially identical argument holds: $ij \in \bar{E}$, and so $\kappa_i \geq \kappa_j \bar{\gamma}_{ij}$, implying that $\gamma^{\mu'}_{ij} \leq 1$.
        There are no arcs in $E \cup \ole{E^\mu}$ leaving $\hat{V}$, so the only remaining case is that $i \in V \setminus \hat{V}$, $j \in \hat{V}$. 
        By the definition of \textsc{Highest-Gain},
        $\kappa_i \leq \bar{\gamma}_{ij}$, from which the claim follows immediately.
    
        \nnote{It's somewhat awkward here that $\mu'$ changes in the algorithm, so it's awkward to refer to the value of $\mu'$ before line 9. In the above I avoid it, but it may be better to just define a $\hat{\mu}_i$ or something for the intermediate value.}
        



        \item \emph{for all $i \in \Vm$, 
            $|b^\mu_i| \leq |b^{\mu'}_i| \leq \lceil |b^\mu_i| \rceil$}:
        This is trivial for $i \in \VZ$.
            Consider any node $i \in \VP \cup \VN$. Since $\theta_i\le
            1$, we have $|b^{\mu'}_i| \geq |b^\mu_i|$.
    Furthermore,
$|b^{\mu'}_i| = |b^\mu_i| / \rndgain_i \leq |b^\mu_i| / \gamma'_{it'} = \lceil |b^\mu_i| \rceil$.
    \end{enumerate}
\end{proof}
\begin{lemma}
    Every node is anchored in the labelling $\mu'$ returned by \textsc{Round-Labels}.
\end{lemma}
\begin{proof}
    First, consider any node $i \in \VP \cup \VN$ for which $\rndgain_i = \gamma'_{it'}$.
    Then $b^{\mu'}_i = b^{\mu}_i \cdot \gamma'_{it'}$ is either the round up or round down of $b^\mu_i$.
    Further, we have
    \[ |b^\mu_i| \leq |\nabla f^\mu_i| + |\nabla f^\mu_i - b^\mu_i| \leq 5n(n-1) + 3n \leq 5n^2, \]
since $f_e \leq 5n$ for all $e \in E$, and $\Ex(f,\mu)$ and $\Deficit(f,\mu)$ are both bounded by $3n$.
    Thus $i$ is an anchor for $\mu'$.

    Now consider any $i \in \hat{V}$, and let $P$ be   a highest gain path to $t'$ in $(V', E')$.
    Then the penultimate vertex of $P$ is an anchor, and all preceding arcs of $P$ are tight with respect to $\mu'$; thus $i$ is anchored in $\mu'$.

    Finally, consider any $i \in V \setminus \hat{V}$, and let $P$ be
    the highest gain path to $\hat V \cup \{ \bar{t} \}$ in $(\bar{V}, \bar{E})$.
    If the  final arc of $P$ is $j\bar{t}$, then $\mu'_j = 1$ and $i$
    is anchored.
Otherwise, $P$ connects $j$ to a node $k\in \hat V$. Either $k$ itself
is an anchor, or it is connected to an anchor $k'\in \hat V$ by a
tight path $P'$; thus, the concatenation of $P$ and $P'$ connects $i$
to an anchor, and the same argument applies as above.
\end{proof}




\section{Phase one: finding a feasible solution}\label{sec:first-phase}

We are now ready to complete the proof of Theorem~\ref{thm:main}. 
Consider an arbitrary input instance
$\mathcal{I}=(V,E,t,\gamma,b)$, without the assumption \nameref{cond:init} on the existence of an initial fitting pair
$(\bar f,\bar \mu)$. 
We present a two-phase algorithm, similar to the two-phase simplex method.
The first phase may conclude infeasibility or unboundedness, and otherwise
returns an initial fitting pair
$(\bar f,\bar \mu)$ on a subset of the original node set. The second
phase solves the flow maximization problem starting from $(\bar f,\bar
\mu)$.
The main work in both phases will be an execution of Algorithm~\ref{alg:roundmaxflow}.
\nnote{This two-phase structure is slightly less explicit than before. Not sure if that's a problem.}\lnote{I think this is OK.}

For a node set $W \subseteq V$ with $t\in W$, we let $\cI[W]$ denote the
restriction of the instance to $W$, obtained by deleting all nodes in
$V\setminus W$ and the incident arcs.

 \modified{We first describe the algorithm for the case when there are no flow generating cycles in the instance. This is extended later to instances with flow generating cycles. Note that since we work with uncapacitated instances, a flow generating cycle can create unlimited amounts of  flow.}

\paragraph{No flow generating cycles.}
\modified{Let us first assume that  $\gamma(C) \leq 1$ for all cycles $C \subseteq E$.}

Let $\Vt$ denote the set of nodes reachable from $t$ (this set can be identified using depth-first search).
We can satisfy all demands in $\Vt$ from $t$.
More precisely, we construct a fitting pair $(\ft, \mut)$ to $\cI[\Vt]$ with $\ft$ feasible as follows. 
Let $\mut_t:=1$, and for every $i\in \Vt\setminus\{t\}$, we define
\begin{equation}\label{def:init-mu-t}
 \mut_i:=\max\{\gamma(P)\colon P \text{ is a directed walk
 from $t$ to $i$}\}.
\end{equation}
This is well-defined and finite since there are no flow generating cycles.%
\footnote{This is called a
  \emph{canonical labelling from $t$} in \cite{Goldberg91}.} 
We can efficiently compute the values $\mut_i$ via the multiplicative
adaptation of Dijkstra's algorithm. 
It is immediate by the definition that $\mut$ is a feasible labelling for $\cI[\Vt]$, 
and that there exists a tight $t$-$i$ path $P_i$ for every $i\in \Vt$ w.r.t.\ $\mut$. 
For every $i\in \Vt\cap \VP$, we let $g^{(i)}$ be the flow sending $b_i$ units on $P_i$ to the node $i$
(that is, the relabelled flow will be $b_i^{\mut}$ on every arc of $P_i$). 
We let $\ft:=\sum_{i\in \Vt\cap \VP} g^{(i)}$.
Clearly, $(\ft,\mut)$ is a fitting pair.

Next, we consider $\Vr := V \setminus \Vt$.
For each $i\in \Vr$, let
\begin{equation}\label{def:init-mu}
 \muri_i:=\max\{\gamma(P)\colon P \text{ is a directed walk ending in } i\},
\end{equation}
with the convention that  $\gamma(P)=1$ for a walk $P$ comprising a single node.
Again, this can be computed via the multiplicative adaptation of Dijkstra's algorithm.
Since there are no flow generating cycles, $\muri$ is well defined and finite,
and $\muri_i\gamma_{ij}\le \muri_j$ for every $ij\in E[\Vr]$.

Now construct the instance 
$\mathcal{I}'=(\Vr \cup \{t'\},E',t',\gamma',b')$ as follows.
Let $E':=E[\Vr] \cup \{t'j\colon j\in \Vr \cap \VP\}$,
$b'_i:=b_i$ for all $i\in \Vr$, and $\gamma'_e:=\gamma_e$ for all $e\in E[\Vr]$. 
For the new arcs $t'j$, we let $\gamma'_{t'j}:=\muri_j$.

Extend $\muri$ to a dual solution for $\mathcal{I}'$ by setting $\muri_{t'} := 1$.
Define a primal solution $\fri$ by setting $\fri_e := 0$ for all arcs $e \in E$, and 
$\fri_{t'j} := b_j / \mu_j$ for the arcs in $\delta^+(t')$.
Note that this flow is feasible, since
$\netflow{\fri}{j}=\max\{b_j,0\}$ for every $j\in \Vr$.
Also, $\fri$ fits $\muri$, since all arcs in $\delta^+(t')$ are
tight.
Thus $(\fri, \muri)$ provides the initial fitting pair for $\instance'$ required by \nameref{cond:init}.

Using this, we can use Algorithm~\ref{alg:roundmaxflow}
to obtain an essentially optimal fitting pair $(\fr, \mur)$ for $\instance'$.

\nnote{I'm tempted to unwrap this lemma.}
\begin{lemma}\label{lem:first-phase-feasible}
The original instance $\cI$ is feasible if and only if
$\fr(\delta^+(t'))=0$.
\end{lemma}
\begin{proof}
Let us first assume $\fr(\delta^+(t'))=0$.
We obtain a feasible flow $\bar f$ in $\cI$ by setting 
$\bar f_e:=\ft_e$ if $e\in E[\Vt]$, $\bar f_e:=\fr_e$ if
$e\in E[\Vr]$, and $\bar f_e:=0$ otherwise. 

In the other direction, assume there is a feasible solution $\bar f$ to $\cI$. 
Then we obtain a feasible solution $f'$ to
$\mathcal{I'}$ by setting $f'_{ij}= \bar f_{ij}$ for $ij\in
E[\Vr]$ and $f'_{t'j}=0$ for all $t'j\in E'$. Note that
$\delta^-(\Vr)=\emptyset$ by the definition of $\Vt$.
Thus, since $\delta^-(t') = \emptyset$ and $\fr$ is optimal, $\fr(\delta^+(t')) \leq f'(\delta^+(t')) = 0$.
\end{proof}

If we determine that $\cI$ is feasible, we are now in a position to provide an initial fitting pair $(\bar f,\bar \mu)$ as follows.
Define $\bar f$ as in the above proof; this is then feasible.
Note that $\delta^-(\Vr)=\emptyset$. 
If there are no arcs from $\Vr$ to $\Vt$,
then we define $\bar \mu_i:=\mut_i$ if $i\in \Vt$
and $\bar \mu_i:=\mur_i$ if $i\in \Vr$. 
Otherwise, we let $\delta:=\max_{i\in \Vr, j\in  \Vt} \mur_i\gamma_{ij}/\mut_j$,
and define  $\bar \mu_i:=\delta\mut_i$ if $i\in \Vt$
and $\bar \mu_i:=\mur_i$ if $i\in \Vr$. 
In either case, it is easy to see that $(\bar f, \bar \mu)$ is a fitting pair satisfying \nameref{cond:init} for $\cI$; ``phase one'' is complete.
We can then proceed with ``phase two'' and apply Algorithm~\ref{alg:roundmaxflow} to obtain an essentially optimal fitting pair.

\paragraph{The general case (flow generating cycles allowed).}
Let us call a node $i\in V$ \emph{flooded} if there exists a flow
generating cycle $C\subseteq E$ (that is, $\gamma(C)>1$), along with a path 
$P\subseteq E$ connecting a node of $C$ to $i$. 
We can use $C$ and $P$ to generate arbitrary amounts of excess flow at node $i$; 
hence arbitrary demand $b_i$ can be met at a flooded node $i$.

We let $\Vg$ denote the set of flooded nodes.
This set can be efficiently identified as follows.
A flow generating cycle is a negative cycle with respect to the cost function
$c_e=-\log\gamma_e$. Hence we can adapt any negative cycle detection
subroutine (e.g. \cite[Chapter 5.5]{amo}) to find a negative cycle, or
conclude that none exists in $O(nm)$ time. We can use a multiplicative adaptation of
the negative cycle detection algorithms to avoid computations with logarithms.
If a flow generating cycle $C$ is found, then we include all nodes
incident to $C$ into $\Vg$, as well as all other nodes that are reachable on a directed path from $C$.
We remove every vertex added to
$\Vg$ from $V$, and repeat the same process. 
In $O(n)$ iterations, we correctly identify $\Vg$.

We can easily construct a vector $\fg \in\R_+^{E[\Vg]}$ with $\nabla \fg_i\ge b_i$ for each $i\in \Vg$ as
follows. For every node $i\in \Vg\cap \VP$, the algorithm identifies a flow
generating cycle $C$ and a path $P$ connecting $C$ to $P$. Let  $f^{(i)}$ be the flow supported on $C\cup P$ such that
$\nabla f^{(i)}_i=b_i$ and $\nabla f^{(i)}_j=0$ for $j\neq
i$. 
We then define $\fg :=\sum_{i\in \Vg\cap \VP} f^{(i)}$.

Since $\Vng := (V \setminus \Vg) \cup \{t\}$ contains no flow generating cycles, we can apply the algorithm described earlier to determine 
whether $\hat{\cI} := \cI[\Vng]$ is feasible, and if it is, to determine an essentially optimal fitting pair $(\hat{f}, \hat{\mu})$.
If $\hat{\cI}$ is not feasible, then (since $\delta^-(W) = \emptyset$) $\cI$ is not feasible either.
If $\hat{\cI}$ is feasible and $t \in \Vg$, then $\cI$ is feasible and unbounded.
So consider the final case, that $\hat{\cI}$ is feasible and $t \notin \Vg$.

We construct an optimal primal-dual pair $(f^*,\mu^*)$ to $\mathcal{I}$ by combining $(\hat{f}, \hat{\mu})$ and $\fg$, as follows. 
Let $f^*_e:=\fg_e$ for
$e\in E[\Vg]$, $f^*_e:=\hat{f}_e$ for $e\in E[\Vng]$, and $f^*_e:=0$ otherwise. 
Let $\mu^*_i:=\hat{\mu}_i$ if $i$ can reach $t$ in $E_{f^*}$, and
$\mu^*_i=\infty$ otherwise (this will include all nodes in $\Vg$).
Then $f^*$ and $\mu^*$ are both feasible, and satisfy complementary slackness; thus $(f^*, \mu^*)$ is an optimal solution.

\paragraph{Running time.} 
The time to compute the set of flooded nodes $\Vg$ is dominated by
$O(n)$ negative cycle detections, yielding a total running time of $O(n^2m)$. 
Computing an essentially optimal solution to an instance without flow generating cycles 
requires two calls to Algorithm~\ref{alg:roundmaxflow}, 
as well as a depth-first search, an application Dijkstra's algorithm,
and determining $O(n)$ path flow values. 
The time of $O(mn(m+n\log n)\log(n^2/m))$ required for the calls to Algorithm~\ref{alg:roundmaxflow} clearly dominates.

\paragraph{Encoding length.} It is straightforward to see that the labels $\mut$
and $\mur$ and the flows $\fg$, $\ft$ and $\fr$ have polynomially bounded encoding length in the
size of the input. 
Since Algorithm~\ref{alg:roundmaxflow} is in PSPACE, it follows that $(\bar f, \bar \mu)$ and hence also $(f^*, \mu^*)$ have polynomially bounded encoding length.

The proof of Theorem~\ref{thm:main} is now complete.





\section{Conclusions and future research}\label{sec:conclusion}

\nnote{We should agree on whether a hyphen or not in minimum-cost (currently we have it differently in the intro). Somewhat prefer without, since I'd say we see "minimum cost flow" much more often than "minimum-cost flow". But either is fine.}

A next milestone towards a strongly polynomial algorithm for general linear programming
is to find a strongly polynomial algorithm for the minimum cost generalized flow problem.
Here, each arc $e$ has in addition a cost $c_e$ per unit of incoming flow;
there is no sink, and instead one wishes for a feasible flow satisfying all demands and of minimum cost.
Generalized flow maximization is a special case of this, either indirectly as solving the feasibility question, 
or directly by setting all costs to be zero except for a negative cost flow-absorbing loop at the sink.

One major apparent obstacle when moving from flow maximization to cost minimization is that the notion of relabelling does not seem to have an analogue for minimum cost generalized flow. For the flow maximization problem, the dual constraints are $\gamma_{ij} y_j-y_i\le 0$ for every edge $(i,j)$, with the convention \modified{$y_t=1$; recall that the labels are the inverses of the dual variables. }These inequalities can be conveniently transformed into the multiplicative form 
\eqref{dual}, naturally giving rise to the notion of fitting pairs, and enabling to work with regular flows when restricted to the set of tight arcs. No such transformation is possible for the more general constraints $\gamma_{ij} y_j-y_i\le c_{ij}$. 
\nnote{Shouldn't it be $y_t = 1$? Should we say that the $y_j$'s are the inverses of the labels?}

As discussed in the introduction, minimum cost generalized flow includes as special cases both generalized flow maximization, and the feasibility problem for a linear program whose constraint matrix has at most two nonzero entries per row. One would expect that a strongly polynomial algorithm for minimum cost generalized flow will treat the primal and dual sides of the problem in an integrated framework. Currently, there is a significant difference between the algorithms for the two sides. For the dual problem, Megiddo's algorithm \cite{megiddo83} heavily relies on parametric search, and this has been a key ingredient also in subsequent algorithms for the problem, e.g., \cite{CM94}.  Therefore, new 
approaches for the dual problem may be necessary to tackle the minimum cost generalized flow problem.

Another possible extension would be to obtain strongly polynomial algorithms for more general problem variants, such as maximum generalized flows with concave gains, studied in \cite{Shigeno06,Vegh11}. Whereas one cannot hope for strongly polynomial algorithms for the general model---the optimal solution may not even be rational---it might still be worth exploring tractable special cases, in the spirit of \cite{Vegh-sep} for minimum cost regular flows with separable convex objectives.

\lnote{wrote it and then noticed that most points were already made in my previous paper :) not sure if this is a problem.}



\paragraph*{Acknowledgements.}
We are very grateful to Jos\'e Correa and Andreas Schulz for many interesting discussions which led to this work.
We thank the anonymous referees for their detailed and constructive comments that helped improve the exposition.
Part of this work was done while the authors were participating in the Hausdorff Trimester on Discrete Mathematics in Fall 2015.

\bibliographystyle{abbrv}
\bibliography{strongly}

\end{document}